\documentclass[garamond]{acdis}
\usepackage{algorithm,algpseudocode,tabularx,diagbox}
\usepackage{caption}
\usepackage{subcaption}
\usepackage{todonotes}
\usepackage{etoolbox}
\usepackage{mathtools} 

\author{Yu-Chen Shen and Matthieu R. Bloch \\[1.0ex]
\normalfont\small School of Electrical and Computer Engineering, Georgia Instfitute of Technology \\[-1.5ex]
\normalfont\small yshen469@gatech.edu, matthieu.bloch@ece.gatech.edu}
\title{Entanglement-Assisted Bosonic MAC: Achievable Rates and Covert Communication}
\leftnote{}
\rightnote{}

\begin{document}
\maketitle
\thispagestyle{fancy}
\allowbreak
\begin{abstract}
We consider the problem of covert communication over the entanglement-assisted (EA) bosonic multiple access channel (MAC). We derive a closed-form achievable rate region for the general EA bosonic MAC using high-order phase-shift keying (PSK) modulation. Specifically, we demonstrate that in the low-photon regime the capacity region collapses into a rectangle, asymptotically matching the point-to-point capacity as multi-user interference vanishes. We also characterize an achievable covert throughput region, showing that entanglement assistance enables an aggregate throughput scaling of $O(\sqrt{n} \log n)$ covert bits with the block length $n$ for both senders, surpassing the square-root law as in the point-to-point case. Our analysis reveals that the joint covertness constraint imposes a linear trade-off between the senders throughput.
\end{abstract}

\section{Introduction}\label{sec:intro}
We consider a bosonic Multiple Access Channel (MAC) $\mathcal{L}^{(\tau, \kappa, N_B)}_{A_XA_Y\to BW}$ where two senders $A_X$ and $A_Y$ attempt to transmit classical information by modulating their respective modes, $\hat{a}_X$ and $\hat{a}_Y$. These modes are first mixed through a linear combination characterized by a transmission parameter $\tau$, resulting in the combined mode $\hat{a} = \sqrt{\tau}\hat{a}_X + \sqrt{1-\tau}\hat{a}_Y$. This merged signal then undergoes a second mixing stage as it propagates through a lossy bosonic channel, modeled as a beamsplitter with transmissivity $\kappa$ that mixes the signal $\hat{a}$ with a thermal environment $\hat{e}$ in a thermal bath $\rho_{N_B}$ with mean photon number $N_B$~\cite{Weedbrook2012Gaussian}. The resulting annihilation operators at the receiver (Bob) and the potential eavesdropper (Willie) are given by \cite{Shi2021Entanglement,Shi_2022}:
\begin{align}
    &\hat{b} = \sqrt{\kappa}\hat{a} + \sqrt{1-\kappa}\hat{e};\\
    &\hat{w} = -\sqrt{1-\kappa}\hat{a} + \sqrt{\kappa}\hat{e}.
\end{align}
The communication is supported by entanglement resources shared between $A_X$, $A_Y$, and Bob. Specifically, sender $A_X$ shares a set of two-mode squeezed vacuum (TMSV) pairs with Bob, as does sender $A_Y$ with Bob, while no entanglement is shared between the two senders~\cite{Shi2021Entanglement}. The TMSV state $|\psi\rangle_{AI}$ is a bipartite entangled state defined in the Fock basis as \cite{Weedbrook2012Gaussian}
\begin{align}
    |\psi\rangle_{AI} = \sum_{k=0}^{\infty} \sqrt{\frac{N_S^k}{(N_S+1)^{k+1}}} \ket{k}_A \ket{k}_I,\label{eq:TMSV}
\end{align}
where $N_S$ is the mean photon number per mode. Operationally, each sender applies a phase modulation $\hat{U}_{\theta} = \exp(j2\theta\hat{a}^{\dagger}\hat{a})$ on its reference system before transmission~\cite{Su2024Achievable}. Bob stores the idler systems $I_X$ and $I_Y$ and performs a joint measurement on the channel output $\hat{b}$ on system $B$ and both idlers to decode the messages.

MACs have already attracted much attention. The capacity region is fully known for classical channels~\cite{Ahlswede_1974} and finite-dimensional classical-quantum channels with and without entanglement assistance~\cite{Winter_2001,Hsieh2008Entanglement}, though the evaluation of the unassisted region is often intractable because of regularization. There has been recent progress in characterizing the capacity region for bosonic Gaussian MACs~\cite{Yen2005Multiple,Shi2021Entanglement,Shi_2022}, including a single letter characterization of the optimal sum-rate for phase-insensitive bosonic Gaussian MACs~\cite{Shi_2022}. Entanglement-sharing between senders shall not be considered in the present work but recent results have unveiled surprising gains~\cite{Leditzky_2020,Pereg_2025}. 

We shall focus on the specific regime of covert communications, in which the goal is to reliably communicate while  also avoiding detection by an eavesdropper~\cite{Bash2015Hiding,Bash2013Limits}. Covert communications are often governed by a \emph{square-root law}, which requires the number of reliable bits to scale with the square root of the block length $n$. One can define a notion of covert capacity, which has been fully characterized for many classical channels~\cite{Bloch2015b,Wang2016b,Bouette_2025,Bounhar2025Capacity}, classical quantum channels~\cite{Wang2016c}, and Gaussian bosonic channels~\cite{Bash2015a,Bullock2020,Gagatsos2020Covert}. In particular, the utilization of shared entanglement allows one to improve the square-root law scaling to $O(\sqrt{n}\log n)$~\cite{Gagatsos2020Covert,Zlotnick2025Entanglement} while still making efficient use of entanglement resources~\cite{Cox2023Transceiver,Wang2024Resource}. 

The two main contributions of the present work are
\begin{inparaenum}[i)]
\item a general operational characterization of the achievable rate region for the entanglement-assisted (EA) Bosonic MAC;
\item its subsequent application to the regime of covert communication.
\end{inparaenum}
In particular, we extend what was already known for classical covert MACs~\cite{Arumugam2018a}.

Our first result (Theorem~\ref{theo:MAC_general}) establishes a closed-form achievable rate region for the EA Bosonic MAC under general power constraints using single-layer, high-order PSK modulation. By decomposing joint mutual information into single-user components, we map the MAC to effective point-to-point channels characterized in~\cite{Su2024Achievable} to yield explicit operational rate forms. Notably, while the general MAC capacity region typically forms a pentagon, the vanishing signal power in the covert regime renders multi-user interference negligible, collapsing our region into a rectangle that asymptotically matches the point-to-point EA capacity limits~\cite{Gagatsos2020Covert} for both senders (Remark~\ref{rem:optimality}).

Our second result (Theorem~\ref{thm:main}) extends the setup to covert communication over MAC with limited entanglement. Adopting the framework from~\cite{Wang2024Resource}, we employ a dual-layer strategy: \textbf{Layer 1} uses key assisted sparse coding to ensure covertness, while \textbf{Layer 2} utilizes EA high-rate PSK modulation to maximize throughput. We show this approach achieves an efficient $O(\sqrt{n}\log n)$ rate scaling. Crucially, we demonstrate that joint covertness enforces a strict linear trade-off between the senders' rates, as they must effectively partition a single global energy budget to remain undetected (Remark~\ref{rem:tradeoffs}).

The primary mathematical challenges addressed here involve generalizing point-to-point results to a two-sender setting, requiring two critical technical adjustments. First, to render the characterization of mutual information tractable, we utilize the chain rule in conjunction with a receiver strategy based on heterodyne detection of the idler modes. This approach effectively decouples the complex two-sender terms into independent point-to-point components. Second, we rigorously establish the codebook existence (Lemma~\ref{lemm:reliable_resolvable}) by addressing both channel resolvability and reliability. Specifically, we apply multipartite convex splitting~\cite{Cheng2023Quantum, George2024Coherent} to ensure the average output state remains statistically indistinguishable from the vacuum (covert) state, while simultaneously invoking joint achievability bounds from~\cite{Sen2018Unions}, incorporating a necessary state truncation at the receiver to guarantee reliability in the continuous-variable regime.

While prior works~\cite{Shi2021Entanglement,Shi_2022,Anderson2022Fundamental} established general capacity bounds for bosonic and entanglement-assisted MACs, closed-form expansions characterizing the rate region boundary in the covert limit have remained elusive. We address this with Theorem~\ref{theo:MAC_general} by deriving a concrete, closed-form lower bound on the capacity region and establishing its explicit asymptotic expansion in the low-power limit. This expansion reveals that the achievable region asymptotically matches the point-to-point EA capacity in the covert regime. Furthermore, regarding covert communication (Theorem~\ref{thm:main}), we generalize the analysis in~\cite{Wang2024Resource} to the MAC setting, explicitly characterizing the resource trade-off between two users governed by the joint covertness constraint.

The paper is organized as follows: Section~\ref{sec:preliminary} reviews preliminaries, Section~\ref{sec:general-MAC} derives the general EA-MAC rate region, and Section~\ref{sec:covert_MAC} characterizes the covert regime and its associated resource trade-offs, with the proof in Section~\ref{sec:covert_proof}.

\section{Notation and Preliminaries}\label{sec:preliminary}
Let $\mathbb{R}_{+}$ and $\mathbb{N}_{*}$ denote the sets of non-negative real numbers and positive integers, respectively. For any set $\Omega$, the indicator function is defined as $1\{\omega \in \Omega\} = 1$ if $\omega \in \Omega$ and $0$ otherwise. A sequence of $n$ elements is denoted by $x^n \triangleq (x_1, \dots, x_n) \in \mathcal{X}^n$, and we denote the type of $x^n$ as $\hat{p}_X$. Throughout this work, $\log$ denotes the natural logarithm (base $e$), and all information quantities are expressed in nats. Let $\mathcal{D}(\mathcal{H})$ denote the set of density operators acting on a separable Hilbert space $\mathcal{H}$, and let $\mathcal{D}_{\le}(\mathcal{H})$ denote the set of subnormalized density operators. For a quantum state $\rho \in \mathcal{D}(\mathcal{H})$, the von Neumann entropy is $\mathbb{H}(\rho) \triangleq -\text{Tr}(\rho \log \rho)$, and the trace distance between two states $\rho$ and $\sigma$ is defined as $\frac{1}{2}\|\rho - \sigma\|_1$, where $\|A\|_1 \triangleq \text{Tr}(\sqrt{A^\dagger A})$. 

For a quantum state $\rho_{AB}$, the Holevo information is ${I}(A;B)_\rho \triangleq {D}(\rho_{AB} \| \rho_A \otimes \rho_B)$. We utilize non-asymptotic entropic quantities, including the hypothesis testing relative entropy
\begin{align}
{D}_H^\epsilon(\rho \| \sigma) \triangleq -\log \inf \{ \text{Tr}(\Pi \sigma) : 0 \le \Pi \le I, \text{Tr}(\Pi \rho) \ge 1-\epsilon \},
\end{align}
and the hypothesis testing mutual information ${I}_H^\epsilon(A;B)_\rho \triangleq {D}_H^\epsilon(\rho_{AB} \| \rho_A \otimes \rho_B)$. The max-relative entropy of $\rho, \sigma \in \mathcal{D}_{\le}(\mathcal{H})$ such that $\text{supp}(\rho) \subseteq \text{supp}(\sigma)$ is defined as ${D}_{\max}(\rho \| \sigma) \triangleq \inf \{ \lambda \in \mathbb{R} : \rho \le e^\lambda \sigma \}$. The $\epsilon$-smooth max-relative entropy is defined as ${D}_{\max}^\epsilon(\rho \| \sigma) \triangleq \inf_{\rho' \in \mathcal{B}^\epsilon(\rho)} {D}_{\max}(\rho' \| \sigma)$, where $\mathcal{B}^\epsilon(\rho) \triangleq \{ \sigma \in \mathcal{D}(\mathcal{H}) : P(\rho, \sigma) \le \epsilon \}$. The smooth max-mutual information is defined as ${I}_{\max}^\epsilon(A;B)_\rho \triangleq {D}_{\max}^\epsilon(\rho_{AB} \| \rho_A \otimes \rho_B)$. The purified distance $P(\rho, \sigma)$ for $\rho, \sigma \in \mathcal{D}_{\le}(\mathcal{H})$ is defined as
\begin{align}
P(\rho, \sigma) \triangleq \sqrt{1 - F^2(\rho \oplus [1-\text{Tr}(\rho)], \sigma \oplus [1-\text{Tr}(\sigma)])},
\end{align}
where $F(\rho, \sigma) \triangleq \|\sqrt{\rho}\sqrt{\sigma}\|_1$ is the fidelity. The second-order asymptotics are governed by the quantum relative entropy variance $V(\rho \| \sigma) \triangleq \text{Tr}(\rho (\log \rho - \log \sigma - {D}(\rho \| \sigma))^2)$.

For the analysis of bosonic systems, we adopt the standard formalism of Gaussian quantum information \cite{Weedbrook2012Gaussian, Serafini2023Quantum}. In short, Gaussian states can be fully characterized by their first moments $\mu$ and covariance matrices $\Lambda$. We omit an exhaustive review of characteristics about Gaussian states and Gaussian channels for brevity, referring readers to \cite{Weedbrook2012Gaussian, Serafini2023Quantum} for comprehensive definitions and standard properties.

To establish the main results of this paper, we require one-shot bounds that characterize the reliability of message decoding and the resolvability of the channel output for covertness. The following lemma extends standard quantum multiple-access channel results to the one-shot regime, providing the necessary conditions for codebook existence.

\begin{lemma} \cite{Cheng2023Quantum, George2024Coherent, Sen2018Unions}\label{lemm:reliable_resolvable}
    Fix $\varepsilon \in (0,1), \eta\in(0,\varepsilon), \delta \in (0,1)$, and $\delta' \in (0,\delta)$. Suppose a c-q state $\rho_{XYA_XA_Y}$ and a channel $\mathcal{G}: \rho_{XYA_XA_Y} \mapsto \rho_{XYZW}$, where $A_X$ and $A_Y$ are mutually independent classical systems, $Z$ is a finite-dimensional quantum system held by Bob (receiver), and $W$ is a quantum system held by Willie (evasdropper). 
    Then there exists a coding scheme utilizing random codebooks $\mathcal{C}_{X_1}$ and $\mathcal{C}_{Y_1}$ (which map message-key pairs to the classical input systems $A_X$ and $A_Y$) with message set sizes $M_{X_1}, M_{Y_1}$ and shared secret key sizes $S_X, S_Y$, such that:
    \begin{align}
        &\log M_{X_1} \geq I_H^{(\varepsilon-\eta)^2}(X:YZ)_\rho - 3 \log_2\left(\frac{46}{3}\eta\right)\label{eq:rel_1}\\
        & \log M_{Y_1} \geq I_H^{(\varepsilon-\eta)^2}(Y:XZ)_\rho - 3 \log_2\left(\frac{46}{3}\eta\right)\label{eq:rel_2}\\
        & \log M_{X_1} + \log M_{Y_1} \geq I_H^{(\varepsilon-\eta)^2}(XY:Z)_\rho - 3 \log_2\left(\frac{46}{3}\eta\right)\label{eq:rel_3}\\
        &\mathbf{E}_{\mathcal{C}_{X_1}, \mathcal{C}_{Y_1}, S_X,S_Y}\left[\mathbf{P}((\hat{M}_{X_1}, \hat{M}_{Y_1}) \neq (M_{X_1},M_{Y_1})|S_X, S_Y)\right] \leq 46\varepsilon, \label{eq:reliability_error}
    \end{align}
    and
    \begin{align}
        &\log M_{X_1}S_X \leq I^{\frac{\delta-\delta'}{6}}_{\max}(X:W)_{\rho} - 2 \log\left(\frac{\delta'}{3}\right)\label{eq:res_1}\\
        &\log M_{Y_1}S_Y \leq I^{\frac{\delta-\delta'}{6}}_{\max}(Y:W)_{\rho} - 2 \log\left(\frac{\delta'}{3}\right)\label{eq:res_2}\\
        & \log M_{X_1}S_X + \log M_{Y_1}S_Y \leq I^{\frac{\delta-\delta'}{12}}_{\max}(XY:W)_{\rho} - 2 \log\left(\frac{\delta'}{6}\right)\label{eq:res_3}\\
        & \mathbf{E}_{\mathcal{C}_{X_1},\mathcal{C}_{Y_1}}\left[\frac{1}{2}\left\|\hat{\rho}_W-\rho_W\right\|_1\right] \leq \delta.\label{eq:covert_in_lemma}
    \end{align}
    Here, $\hat{\rho}_W$ denotes the induced soft-covering state on the system $W$, defined as the uniform mixture of the channel outputs over the codewords in the chosen codebooks $\mathcal{C}_{X_1}$ and $\mathcal{C}_{Y_1}$:
    \begin{align}
        \hat{\rho}_W = \frac{1}{M_{X_1}S_{X}M_{Y_1}S_Y} \sum_{m_{X_1}, k_X} \sum_{m_{Y_1},k_Y} \rho_W^{x_{m_{X_1}, k_X}, y_{m_{Y_1}, k_Y}},\label{eq:soft_covering state}
    \end{align}
    where $\rho_W^{x,y}$ corresponds to the state of the quantum side information on system $W$ given the input pair $(x,y)$.
\end{lemma}
\begin{proof}[Proof of Lemma~\ref{lemm:reliable_resolvable}]
    We utilize random codebooks $\mathcal{C}_{X_1}$ and $\mathcal{C}_{Y_1}$ where each codeword corresponding to indices $(m_{X_1}, k_X)$ is generated i.i.d. according to the distribution $p_X(x)$ induced by $\rho_X$ (and similarly for $Y$).
    
    For reliability, we analyze the decoding performance conditioned on the shared keys $k_X$ and $k_Y$ being known to the receiver Bob. This reduces the problem to a standard multiple-access channel scenario. Invoking the one-shot packing bounds from \cite[Theorem 2']{Sen2018Unions} for the state $\rho_{XYZ}$, if the message rates satisfy the conditions \eqref{eq:rel_1}, \eqref{eq:rel_2}, and \eqref{eq:rel_3} in the lemma, the error probability averaged over random codebooks for fixed keys is bounded as:
    \begin{align}
        \mathbf{E}_{\mathcal{C}_{X_1}, \mathcal{C}_{Y_1}|k_X,k_Y}\left[\mathbf{P}((\hat{M}_{X_1}, \hat{M}_{Y_1}) \neq (M_{X_1},M_{Y_1})|k_X, k_Y)\right] \leq 46\varepsilon.
    \end{align}
    The final bound \eqref{eq:reliability_error} is obtained by averaging this inequality over all possible key realizations.

    For resolvability, we apply the multipartite convex splitting lemma \cite{Cheng2023Quantum, George2024Coherent}. We utilize the fact that the average error in channel resolvability corresponds to the trace distance bound in convex splitting. Following the derivation in \cite[Lemma 12 and Corollary 14]{George2024Coherent}, if the total rates (message plus randomness) satisfy conditions \eqref{eq:res_1}, \eqref{eq:res_2}, and \eqref{eq:res_3} in the lemma, the expected trace distance is bounded as:
    \begin{align}
        \mathbf{E}_{\mathcal{C}_{X_1},\mathcal{C}_{Y_1}}\left[\frac{1}{2}\left\|\hat{\rho}_W-\rho_W\right\|_1\right] \leq \delta.
    \end{align}
    Combining the conditions for reliability and resolvability completes the proof.
\end{proof}

\section{High-rate PSK over bosonic MAC}\label{sec:general-MAC}
In this section, we derive the achievable rate region for the entanglement-assisted (EA) bosonic multiple access channel (MAC) using phase-modulated two-mode squeezed vacuum (TMSV) states. 

To formalize the operational performance analysis, we define the code for the two-sender MAC as follows:
\begin{definition}[EA-MAC Code, Reliability, and Rates]
An $(n, M_X, M_Y, \epsilon)$ entanglement-assisted (EA) code for the bosonic MAC consists of two message sets $\mathcal{M}_X = \{1, \dots, M_X\}$ and $\mathcal{M}_Y = \{1, \dots, M_Y\}$, two encoding maps that modulate the reference systems of independent TMSV pairs, and a joint decoding POVM $\{\Lambda_{m_X, m_Y}\}$ acting on the $n$-fold channel output and stored idler systems $B^n I_X^n I_Y^n$. 
The code is said to be $\epsilon$-reliable if the average probability of error satisfies:
\begin{equation}
    P_e = 1 - \frac{1}{M_X M_Y} \sum_{m_X\in \mathcal{M}_X, m_Y\in\mathcal{M}_Y} \text{Tr}\left( \Lambda_{m_X, m_Y} \sigma_{B^n I_X^n I_Y^n}^{m_X, m_Y} \right) \le \epsilon.
\end{equation}
The transmission rates $R_X$ and $R_Y$ (in nats per channel use) are defined as:
\begin{equation}
    R_X = \frac{\log M_X}{n}, \quad R_Y = \frac{\log M_Y}{n}.
\end{equation}
A rate pair $(R_X, R_Y)$ is achievable if for any $\epsilon > 0$ and sufficiently large $n$, there exists an $(n, M_X, M_Y, \epsilon)$ code.
\end{definition}

Let the modulation order $L_n = 2^n$ be determined by the block length $n$. We define $\Theta_X$ and $\Theta_Y$ as uniform random variables taking values in the $L_n$-PSK constellation $\mathcal{S}_n \triangleq \{2\pi k / L_n\}_{k=0}^{L_n-1}$. These variables modulate the signal modes $A_X$ and $A_Y$ via the unitary phase rotation operator $\hat{U}_{\theta} \triangleq \exp(j2\theta\hat{a}^\dagger \hat{a})$.
In the single-shot regime, the corresponding joint encoded state is given by:
\begin{align}
    &\sigma^{\theta, \phi}_{A_XI_XA_YI_Y} \notag\\
    &= (\hat{U}_{\theta})_{A_X}\otimes (\hat{U}_{\phi})_{A_Y} (\ket{\psi}\bra{\psi}_{A_XI_X}\otimes\ket{\psi}\bra{\psi}_{A_YI_Y})(\hat{U}^\dagger_{\theta})_{A_X}\otimes (\hat{U}^\dagger_{\phi})_{A_Y},
\end{align}
where $\ket{\psi}_{AI}$ is the TMSV state given in \eqref{eq:TMSV}.
The signal modes $A_X$ and $A_Y$ are transmitted through the bosonic MAC channel $\mathcal{L}^{(\tau, \kappa, N_B)}_{A_XA_Y\to BW}$ defined in the introduction. The output state is denoted $\sigma^{\theta, \phi}_{BWI_XI_Y}$. The resulting classical-quantum state, correlating the classical phase registers with the quantum channel output, is defined as:
\begin{align}
    \sigma_{\Theta_X\Theta_YBWI_XI_Y} = \frac{1}{|\Theta_{X}||\Theta_{Y}|}\sum_{\theta\in \Theta_{X}} \sum_{\phi\in \Theta_{Y}} \ket{\theta}\bra{\theta}\otimes \ket{\phi} \bra{\phi} \otimes \sigma^{\theta, \phi}_{BWI_XI_Y}.\label{eq:cq_sigma}
\end{align}
Generalizing to the $n$-shot regime utilizing the channel $(\mathcal{L}^{(\tau, \kappa, N_B)}_{A_XA_Y\to BW})^{\otimes n}$, where senders $A_X$ and $A_Y$ independently modulate their signals via $\Theta_X^{\otimes n}$ and $\Theta_Y^{\otimes n}$, the resulting joint system is described by the $n$-fold tensor product state $\sigma^{\otimes n}_{\Theta_X\Theta_YBWI_XI_Y}$.

To obtain the achievable rate region, we apply the reliability part of Lemma~\ref{lemm:reliable_resolvable}. Since this lemma requires Bob's system (which is $B I_X I_Y$ in our case) to reside in a finite-dimensional Hilbert space, we let Bob apply the following truncation channel:
\begin{align}
    \mathcal{N}_{B^nI_X^nI_Y^n\to Z^nJ_X^nJ_Y^n} = \mathcal{N}_{B\to Z}^{\otimes n}\otimes \mathcal{N}^{\otimes n}_{I_X\to J_X}\otimes \mathcal{N}^{\otimes n}_{I_Y\to J_Y}.\label{eq:projection_channel_1}
\end{align}
Each component channel $\mathcal{N}_{K\to K'}$ is defined as a truncation map via a projection $P_{n,K}$ onto a finite-dimensional subspace for the pairs $(K, K') \in \{(B,Z), (I_X, J_X), (I_Y, J_Y)\}$. Any probability mass residing in the orthogonal complement is redirected to the vacuum state:
\begin{align}
    \mathcal{N}_{K\to K'}(\rho) = P_{n,K}\rho P_{n,K} + \mathrm{Tr}[(I - P_{n,K})\rho] \ket{0}\bra{0}_{K'}.
\end{align}
Composing the MAC channel $(\mathcal{L}^{(\tau, \kappa, N_B)}_{A_XA_Y\to B}\otimes \operatorname{id}_{I_XI_Y})^{\otimes n}$ with $\mathcal{N}_{B^nI_X^nI_Y^n\to Z^nJ_X^nJ_Y^n}$ yields the effective c-q state:
\begin{align}
    \sigma^{\otimes n}_{\Theta_X\Theta_YZJ_XJ_Y} = \mathcal{N}_{B^nI^n_XI^n_Y\to Z^nJ^n_XJ^n_Y}(\sigma^{\otimes n}_{\Theta_X\Theta_YBI_XI_Y}).
\end{align}
For each $n$, we choose the projections $P_{n,B}, P_{n,I_X}, P_{n, I_Y}$ with sufficient rank such that
\begin{align}
    \frac{1}{2}\left\|\sigma^{\otimes n}_{\Theta_X\Theta_YZJ_XJ_Y} - \sigma^{\otimes n}_{\Theta_X\Theta_YBI_XI_Y}\right\|_1 \leq \delta_n,
\end{align}
where $\delta_n$ is chosen small enough to ensure the validity of the subsequent second-order expansion.

Applying the reliability part of Lemma~\ref{lemm:reliable_resolvable}, there exists a random codebooks whose message sizes satisfy the following one-shot lower bounds: for any $\eta \in (0, \varepsilon)$,
\begin{align}
    \log M_X &\geq I_H^{(\varepsilon-\eta)^2}(\Theta_X^n:\Theta_Y^nZ^nJ^n_XJ^n_Y)_{\sigma^{\otimes n}} - 3 \log_2\left(\frac{46}{3}\eta\right)\\
    \log M_Y &\geq I_H^{(\varepsilon-\eta)^2}(\Theta_Y^n:\Theta_X^nZ^nJ^n_XJ^n_Y)_{\sigma^{\otimes n}} - 3 \log_2\left(\frac{46}{3}\eta\right)\\
    \log (M_X M_Y) &\geq I_H^{(\varepsilon-\eta)^2}(\Theta_X^n\Theta_Y^n:Z^nJ^n_XJ^n_Y)_{\sigma^{\otimes n}} - 3 \log_2\left(\frac{46}{3}\eta\right),
\end{align}
with the total error probability satisfying $P_e \leq 46\varepsilon$. By choosing $\eta = {O}(n^{-1/2})$ and noting that the projection error $\delta_n$ vanishes asymptotically, we apply the expansion from \cite[Proposition~13]{Oskouei:2018oeu} (utilizing the result analogous to that detailed subsequently in Lemma~\ref{lemm:expansion}) to yield:
\begin{align}
   \log M_X  &\geq nI(\Theta_{X}:BI_XI_Y\Theta_{Y})_{\sigma} \nonumber\\
   &- \sqrt{nV(\sigma_{\Theta_X\Theta_YBI_XI_Y}\|\sigma_{\Theta_X}\otimes \sigma_{\Theta_YBI_XI_Y})}Q^{-1}(\varepsilon^2)+ O(\log n)\\
    \log M_Y  &\geq nI(\Theta_{Y}:BI_XI_Y\Theta_{X})_{\sigma} \nonumber\\
        &- \sqrt{nV(\sigma_{\Theta_X\Theta_YBI_XI_Y}\|\sigma_{\Theta_Y}\otimes \sigma_{\Theta_XBI_XI_Y})}Q^{-1}(\varepsilon^2)+ O(\log n)\\
     \log M_X M_Y &\geq nI(\Theta_{X}\Theta_{Y}:BI_XI_Y)_{\sigma} \nonumber\\
        & - \sqrt{nV(\sigma_{\Theta_X\Theta_YBI_XI_Y}\|\sigma_{\Theta_X\Theta_Y}\otimes \sigma_{BI_XI_Y})}Q^{-1}(\varepsilon^2)+ O(\log n). \label{eq:second-order-expansion}
\end{align}
In the asymptotic limit ($n \to \infty$), the second-order and error terms are negligible. Consequently, for a fixed signal power $N_S$, an achievable rate region is:
\begin{align}
    \left\{ (R_X, R_Y) \;\middle|\; 
    \begin{aligned}
        R_X &\leq I(\Theta_{X}:BI_XI_Y\Theta_{Y})_{\sigma} \\[1ex]
        R_Y &\leq I(\Theta_{Y}:BI_XI_Y\Theta_{X})_{\sigma} \\[1ex]
        R_X + R_Y &\leq I(\Theta_{X}\Theta_{Y}:BI_XI_Y)_{\sigma}
    \end{aligned} 
    \right\}.
\end{align}
Our next task is to evaluate the explicit expressions for these mutual information terms.

Consider the state in \eqref{eq:cq_sigma}. For given phase angles $\theta$ and $\phi$, the covariance matrix of the receiver's zero-mean conditional Gaussian state $\sigma^{\theta, \phi}_{BI_XI_Y}$ is given by
\begin{align}
\Lambda_{BI_XI_Y}^{\theta,\phi} = \begin{pmatrix} 
(S\kappa + N(1-\kappa))\mathbb{I}_2 & C_q \sqrt{\kappa\tau}R(\theta) & C_q \sqrt{\kappa(1-\tau)}R(\phi) \\
C_q \sqrt{\kappa\tau}R(\theta) & S \mathbb{I}_2 & 0_2  \\
C_q \sqrt{\kappa(1-\tau)}R(\phi) & 0_2 & S \mathbb{I}_2
\end{pmatrix},
\end{align}
where $S = 2N_S + 1$, $N = 2N_T/(1-\kappa) + 1$, with $N_T = (1-\kappa)N_B$, the correlation coefficient is $C_q = 2 \sqrt{N_S(N_S+1)}$, and the rotation matrix $R(\theta)$ is defined as
\begin{align}
    R(\theta) = \begin{pmatrix} 
\cos\theta & \sin\theta \\ 
\sin\theta & -\cos\theta 
\end{pmatrix}.
\end{align}

We begin by analyzing the sum-rate bound $I(\Theta_X\Theta_Y:BI_XI_Y)_{\sigma}$. Let $M_{I_X}$ denote the classical outcomes obtained by performing heterodyne detection on the respective idler modes $I_X$. Using the chain rule for mutual information, we obtain the following expansion:
\begin{align}
    &I(\Theta_X\Theta_Y:BI_XI_Y)_{\sigma} \notag \\
    & = I(\Theta_X:BI_XI_Y)_{\sigma} + I(\Theta_Y:BI_XI_Y|\Theta_X)_{\sigma} \notag \\
    & \geq I(\Theta_X:BI_X)_{\sigma} + I(\Theta_Y:BM_{I_X}I_Y|\Theta_X)_{\sigma} \notag \\ 
    & = I(\Theta_X:BI_X)_{\sigma} + \underbrace{I(\Theta_Y:M_{I_X}|\Theta_X)_\sigma}_{=0} + I(\Theta_Y:BI_Y|M_{I_X}\Theta_X)_{\sigma} \notag \\
    & = I(\Theta_X:BI_X)_{\sigma} + I(\Theta_Y:BI_Y|M_{I_X}\Theta_X)_{\sigma}.\label{eq:two_to_one}
\end{align}
The inequality follows from the data processing inequality. Specifically, in the first term, we discard the system $I_Y$, while in the second term, we perform local heterodyne detection on $I_X$ to obtain $M_{I_X}$. Since partial trace and local measurements constitute local quantum channels that cannot increase mutual information, this yields a valid lower bound. Finally, the term $I(\Theta_Y:M_{I_X}|\Theta_X)$ vanishes due to the independence of $\Theta_Y$ from $\{\Theta_X, I_X\}$.
By symmetry, the analogous decomposition holds for the reverse decoding order. Let $M_{I_Y}$ denote the classical outcomes obtained by performing heterodyne detection on the idler modes $I_Y$. Then:
\begin{align}
     I(\Theta_X\Theta_Y:BI_XI_Y)_{\sigma} \geq I(\Theta_Y:BI_Y)_{\sigma} + I(\Theta_X:BI_X| M_{I_Y}\Theta_Y)_{\sigma}.\label{eq:two_to_one_2}
\end{align}
Our objective is to compute the specific information quantities that govern the achievable rate region derived in \eqref{eq:two_to_one}, and \eqref{eq:two_to_one_2}.

We begin by evaluating the conditional mutual information $I(\Theta_X:BI_X|M_{I_Y}\Theta_Y)_{\sigma}$ in \eqref{eq:two_to_one_2}. To do so, we examine the system properties when $\Theta_Y=\phi$ is fixed and the idler outcome $M_{I_Y}= [q_{I_Y}, p_{I_Y}]^T= \mathbf{m}_{I_Y}$ is obtained. We apply the general Gaussian filtering relations derived in \cite[Eqs.~5.143--5.144]{Serafini2023Quantum}, setting the measurement covariance to $\boldsymbol{\sigma}_m = \mathbb{I}_2$ to represent ideal heterodyne detection. The conditional mean of $BI_X$ becomes
\begin{align}
    \mu^\theta_{BI_X|\mathbf{m}_{I_Y}, \phi}& = \mathbf{E}[BI_X|M_{I_Y} =  \mathbf{m}_{I_Y}, \Theta_Y = \phi] \nonumber\\
    &= \begin{pmatrix} 
C_q\sqrt{\kappa(1-\tau)}\frac{1}{S+1}R(\phi)\mathbf{m}_{I_Y} \\ 
\mathbf{0}_2
\end{pmatrix} \nonumber\\
&= \begin{pmatrix} 
K R(\phi)\mathbf{m}_{I_Y} \\ 
\mathbf{0}_2
\end{pmatrix},\label{eq:condi-mean-1}
\end{align}
where $K =C_q\sqrt{\kappa(1-\tau)}/(S+1) = \frac{2\sqrt{\kappa(1-\tau)N_S(N_S+1)}}{2N_S+2}$. The corresponding conditional covariance matrix is obtained via the Schur complement:
\begin{align}
    \Lambda_{BI_X|\mathbf{m}_{I_Y}, \phi}^{\theta} &= \begin{pmatrix} 
(S\kappa + N(1-\kappa))\mathbb{I}_2 & C_q \sqrt{\kappa\tau}R(\theta)  \\
C_q \sqrt{\kappa\tau}R(\theta) & S \mathbb{I}_2 
\end{pmatrix} \nonumber\\
&\quad - \begin{pmatrix}
    C_q\sqrt{\kappa(1-\tau)}R(\phi)\\
    0_2
\end{pmatrix} (\mathbb{I}_2 + S\mathbb{I}_2)^{-1} \begin{pmatrix}
    C_q\sqrt{\kappa(1-\tau)}R(\phi) &
    0_2
\end{pmatrix}\nonumber\\
& = \begin{pmatrix} 
\left(S\kappa + N(1-\kappa) - \frac{C_q^2\kappa(1-\tau)}{S+1}\right)\mathbb{I}_2 & C_q \sqrt{\kappa\tau}R(\theta)  \\
C_q \sqrt{\kappa\tau}R(\theta) & S \mathbb{I}_2 
\end{pmatrix}. \label{eq:condi-var-1}
\end{align}
\begin{remark}[Variance Penalty of Heterodyne Detection]
Heterodyne detection constitutes a simultaneous measurement of non-commuting quadratures, which, by the uncertainty principle, introduces one unit of vacuum noise \cite{Serafini2023Quantum}. This effectively inflates the variance of the measurement outcomes from $S$ to $S+1$. This physical penalty manifests explicitly in the factor $\frac{1}{S+1}$ appearing in both the attenuation coefficient $K$ of the conditional mean \eqref{eq:condi-mean-1} and the subtraction term of the Schur complement in \eqref{eq:condi-var-1}. This ensures the derived lower bound is physically realizable.
\end{remark}

Importantly, while the conditional mean $\mu^\theta_{BI_X|\mathbf{m}_{I_Y}, \phi}$ depends on the specific realization of systems $M_{I_Y}$ and $\Theta_Y$, the conditional covariance matrix $\Lambda_{BI_X|\mathbf{m}_{I_Y}, \phi}^{\theta}$ is independent of these values. Since the von Neumann entropy of a Gaussian state is determined solely by its covariance matrix, the specific displacement caused by $\mathbf{m}_{I_Y}$ does not affect the mutual information. Therefore, by the definition of conditional mutual information, we can introduce an effective zero-mean $2^n$-PSK classical-quantum state, denoted $\sigma'_{\Theta_XBI_X}$, characterized by the covariance matrix in \eqref{eq:condi-var-1}, such that:
\begin{align}
    I(\Theta_X:BI_X|M_{I_Y}\Theta_Y)_{\sigma} = I(\Theta_X:BI_X)_{\sigma'}.\label{eq:sigma'-I}
\end{align}

Similarly, to evaluate the term $I(\Theta_Y:BI_Y|M_{I_X}\Theta_X)_\sigma$ in \eqref{eq:two_to_one}, we observe that conditioning on $M_{I_X}$ and $\Theta_X$ shifts the mean but leaves the covariance invariant. We can therefore introduce an effective zero-mean $2^n$-PSK c-q state, denoted $\sigma'_{\Theta_YBI_Y}$, characterized by the conditional covariance matrix:
\begin{align}
    \Lambda_{BI_Y|\mathbf{m}_{I_X}, \theta}^{\phi}
    &= \begin{pmatrix} 
    \left(S\kappa + N(1-\kappa) - \frac{C_q^2\kappa\tau}{S+1}\right)\mathbb{I}_2 & C_q \sqrt{\kappa(1-\tau)}R(\phi)  \\
    C_q \sqrt{\kappa(1-\tau)}R(\phi) & S \mathbb{I}_2 
    \end{pmatrix}. \label{eq:condi-var-2}
\end{align}
Consequently, the conditional mutual information reduces to the mutual information of this effective state:
\begin{align}
    I(\Theta_Y:BI_Y|M_{I_X}\Theta_X)_{\sigma} = I(\Theta_Y:BI_Y)_{\sigma'}.\label{eq:sigma'-II}
\end{align}

On the other hand, the unconditional term $I(\Theta_X:BI_X)_{\sigma}$ in \eqref{eq:two_to_one} corresponds to the effective point-to-point link described by the marginal state $\sigma_{\Theta_XBI_X}$ of \eqref{eq:cq_sigma}, in which the side information $\sigma_{BI_X}^{\theta}$ is characterized by the conditional covariance matrix:
\begin{align}
    \Lambda_{BI_X}^{\theta} = \begin{pmatrix} 
(S\kappa + N(1-\kappa))\mathbb{I}_2 & C_q \sqrt{\kappa\tau}R(\theta)  \\
C_q \sqrt{\kappa\tau}R(\theta) & S \mathbb{I}_2 
\end{pmatrix}.
\end{align}
Similarly, $I(\Theta_Y:BI_Y)_{\sigma}$ in \eqref{eq:two_to_one_2} is governed by the marginal state $\sigma_{\Theta_YBI_Y}$, in which the side information $\sigma_{BI_Y}^{\phi}$ is characterized by the covariance matrix:
\begin{align}
    \Lambda_{BI_Y}^{\phi} = \begin{pmatrix} 
(S\kappa + N(1-\kappa))\mathbb{I}_2 & C_q \sqrt{\kappa(1-\tau)}R(\phi)  \\
C_q \sqrt{\kappa(1-\tau)}R(\phi) & S \mathbb{I}_2 
\end{pmatrix}.
\end{align}

Combining the calculations above, we establish the following lemma characterizing the achievable sum-rate.
\begin{lemma}[Achievable Sum-Rate Lower Bound]\label{lemm:additive}
The joint mutual information of the entanglement-assisted bosonic MAC, defined by the c-q state $\sigma$ in \eqref{eq:cq_sigma}, satisfies the following achievable lower bound:
\begin{align}
    &I(\Theta_X \Theta_Y : B I_X I_Y)_{\sigma} \geq \notag \\
    & \max \Big\{ 
    I(\Theta_X : BI_X )_{\sigma} + I(\Theta_Y : BI_Y )_{\sigma'}, 
    I(\Theta_Y : BI_Y )_{\sigma} + I(\Theta_X : BI_X )_{\sigma'} 
    \Big\}. \label{eq:additive_loose}
\end{align}
\end{lemma}

We next evaluate the bounds for the individual rates $R_X$ and $R_Y$. For the rate $R_Y$, the corresponding mutual information is $I(\Theta_Y:BI_XI_Y\Theta_X)_{\sigma}$. Utilizing the data processing inequality, the mutual independence of the inputs ($I(\Theta_Y:\Theta_X)_{\sigma}=0$), and the independence of  $\{\Theta_X,I_X\}$ from the input $\Theta_Y$ ($I(\Theta_Y:I_X|\Theta_X)_{\sigma}=0$), the chain rule expansion simplifies as follows:
\begin{align}
    I(\Theta_Y:BI_XI_Y\Theta_X)_{\sigma} & \geq I(\Theta_Y:BM_{I_X}I_Y\Theta_X)_{\sigma}\nonumber\\
    &= \underbrace{I(\Theta_Y:\Theta_X)_{\sigma}}_{0} + I(\Theta_Y:BM_{I_X}I_Y|\Theta_X)_{\sigma} \nonumber\\
    &= \underbrace{I(\Theta_Y:M_{I_X}|\Theta_X)_{\sigma}}_{0} + I(\Theta_Y:BI_Y|M_{I_X}\Theta_X)_{\sigma} \nonumber\\
    &= I(\Theta_Y:BI_Y)_{\sigma'}, \label{eq:I-2}
\end{align}
where the final equality follows from the definition of the effective state $\sigma'$ in \eqref{eq:sigma'-II}. By symmetry, the achievable rate for $R_X$ is similarly bounded:
\begin{align}
    I(\Theta_X:BI_XI_Y\Theta_Y)_{\sigma} 
    \geq I(\Theta_X:BI_X|M_{I_Y}\Theta_Y)_{\sigma} 
    = I(\Theta_X:BI_X)_{\sigma'}. \label{eq:I-3}
\end{align}
where $\sigma'$ here refers to the corresponding effective state defined in \eqref{eq:sigma'-I}.

The mutual information terms appearing in the bounds \eqref{eq:additive_loose}, \eqref{eq:I-2}, and \eqref{eq:I-3}, namely $I(\Theta_X : BI_X)_{\sigma}$, $I(\Theta_Y : BI_Y)_{\sigma}$, $I(\Theta_Y:BI_Y)_{\sigma'}$, and $I(\Theta_X:BI_X)_{\sigma'}$, can all be interpreted as the Holevo information of specific point-to-point channels. These quantities directly correspond to the results derived in \cite[Proposition 1, 2]{Su2024Achievable}, subject to slight modifications of the covariance matrices. To demonstrate this correspondence explicitly, we restate the relevant propositions below.
\begin{proposition}[Proposition $1$ of \cite{Su2024Achievable}]\label{prop:info_1}
Consider the state 
\begin{align}
    \sigma_{\Theta BI} = \frac{1}{|\Theta|}\sum_{\theta\in \Theta} \ket{\theta}\bra{\theta}\otimes \sigma^{\theta}_{BI}.\notag,
\end{align}
where $\sigma^{\theta}_{BI}$ has mean $0$ and the covariance matrix
\begin{align}
\Lambda_{BBII}^{\theta} = \begin{pmatrix} 
(S\kappa + N(1-\kappa))\mathbb{I}_2 & C_q \sqrt{\kappa}R(\theta)  \\
C_q \sqrt{\kappa}R(\theta) & S \mathbb{I}_2 
\end{pmatrix},\notag
\end{align}
where $S, N, C_q, R(\theta)$ is defined before. Let $N_T = (1-\kappa)N_B$.
For point-to point bosonic channel, under the condition $N_{T}> \max\{\kappa N_{S}-1, \frac{-(1+2\kappa N_{S})+\sqrt{4\kappa N_{S}^{2}+4\kappa N_{S}+1}}{2}\}$, the mutual information of an $2^n$-PSK modulated TMSV state converges to that of a continuously phase-modulated TMSV state:
\begin{align}
&I(\Theta : BI)_\sigma \nonumber\\
&\ge I(\Theta : BI)_{\bar\sigma} - \frac{d(\frac{c^{2}}{1-a-b+ab-abz})^{2^{n}}}{(1-a-b+ab-abz)^{3}(1-(\frac{c^{2}}{1-a-b+ab-abz})^{2^{n}})^{5}} \nonumber\\
&\quad \cdot\sum_{i=0}^{4} P_{i} \left(\frac{c^{2}}{1-a-b+ab-abz}\right)^{i2^{n}},\notag
\end{align}
where $a, b, c, d, z$ are channel-dependent parameters detailed in \cite[Proposition 1]{Su2024Achievable}, $P_i$ are polynomial functions of the constellation size $L_n= 2^n$, and $\bar{\sigma}$ denote the state from uniform phase modulation with the same quantum side information.
\end{proposition}

\begin{proposition}[Proposition $2$ of \cite{Su2024Achievable}]\label{prop:info_2}
Let $N_T = (1-\kappa)N_B$. The achievable rate (mutual information) for a uniform phase-modulated TMSV state $\bar\sigma$ over a lossy thermal-noise bosonic channel is lower bounded by:
\begin{align}
I(\Theta : BI)_{\bar{\sigma}} \geq & \log\left[\frac{(N_{S}+1)N_{T}(N_{T} -\kappa+1)}{N_{T} -\kappa}\right] \nonumber \\
& + N_{S}\log\frac{N_{S}+1}{N_{S}} + N_{S}\log\frac{N_{T} }{N_{T} -\kappa} \nonumber \\
& + (\kappa N_{S}+N_{T} )\log\frac{N_{T} -\kappa+1}{N_{T}-\kappa} \nonumber \\
& + \dots - g(\mu_{+}-\tfrac{1}{2}) - g(\mu_{-}-\tfrac{1}{2}),\notag
\end{align}
where $g(x) = (x+1)\log (x+1) - x\log x$ is the von Neumann entropy of a thermal state and $\mu_{\pm}$ are the symplectic eigenvalues of the conditional Gaussian state:
\begin{align}
\mu_{\pm} = \frac{1}{2} \Bigg[ & \sqrt{(N_{T}  + (1+\kappa)N_{S} + 1)^{2} - 4\kappa N_{S}(N_{S}+1)} \nonumber \\
& \pm (N_{T} + (\kappa - 1)N_{S}) \Bigg].\notag
\end{align}
The omitted terms represented by the ellipsis correspond to the hypergeometric-based lower bounds detailed in \cite[Proposition 2]{Su2024Achievable}. 
\end{proposition}

To evaluate the achievable rates using Propositions~\ref{prop:info_1} and \ref{prop:info_2}, we map the general channel parameters $(\kappa, N_T)$ to effective parameters $(\kappa_{\text{eff}},(N_T)_{\text{eff}})$ for each decoding scenario in Table~\ref{tab:effective_params} by matching the structure of the corresponding covariance matrices.
\begin{table}[h!]
\centering
\caption{Effective Channel Parameters for Propositions~\ref{prop:info_1} and \ref{prop:info_2}}
\label{tab:effective_params}
\renewcommand{\arraystretch}{1.5}
\begin{tabular}{l|c|c}
\hline
\textbf{Information} & $\kappa_{\text{eff}}$ & $(N_T)_{\text{eff}}$ \\ \hline
$I(\Theta_X : BI_X)_{\sigma}$ & $\kappa\tau$ & $N_T + \kappa N_S(1-\tau)$ \\ 
$I(\Theta_Y : BI_Y)_{\sigma}$ & $\kappa(1-\tau)$ & $N_T + \kappa N_S\tau$ \\ \hline
$I(\Theta_X:BI_X)_{\sigma'}$ & $\kappa\tau$ & $N_T + \kappa N_S(1-\tau)\big(1 - \tfrac{4(1+N_S)}{2+2N_S}\big)$ \\ 
$I(\Theta_Y:BI_Y)_{\sigma'}$ & $\kappa(1-\tau)$ & $N_T + \kappa N_S\tau\big(1 - \tfrac{4(1+N_S)}{2+2N_S}\big)$ \\ 
\hline
\end{tabular}
\end{table}
Building upon the inequalities derived in \eqref{eq:additive_loose}, \eqref{eq:I-2}, and \eqref{eq:I-3}, we now formally write down the achievable rate region for the entanglement-assisted bosonic MAC.
\begin{theorem}[EA-MAC Achievable Rate Region]\label{theo:MAC_general}
For the entanglement-assisted bosonic multiple access channel, an achievable rate region for pairs $(R_X, R_Y)$ is defined by the intersection of the following constraints:
\begin{align}
    R_X &\leq I(\Theta_X : BI_X)_{\sigma'}, \notag\\
    R_Y &\leq I(\Theta_Y : BI_Y)_{\sigma'}, \notag\\
    R_X + R_Y &\leq \max \Big\{ 
    I(\Theta_X : BI_X )_{\sigma} + I(\Theta_Y : BI_Y )_{\sigma'},\nonumber\\
    &\quad \quad 
    I(\Theta_Y : BI_Y )_{\sigma} + I(\Theta_X : BI_X )_{\sigma'}
    \Big\}.\notag
\end{align}
The mutual information terms corresponding to the state ($\sigma'$) and the unconditioned state ($\sigma$) are evaluated using the lower bounds in Proposition~\ref{prop:info_1} and Proposition~\ref{prop:info_2} (subject to their respective mild conditions) using the effective channel parameters listed in Table~\ref{tab:effective_params}.
\end{theorem}

In the context of covert communication, the transmitted signal power is constrained by the blocklength $n$, specifically satisfying $N_S = s_n \in \omega(n^{-1/2})\cap o(1)$. By applying Theorem~\ref{theo:MAC_general} together with the non-asymptotic bounds from Propositions~\ref{prop:info_1} and \ref{prop:info_2}, we establish the following asymptotic scaling of the mutual information terms and achievable rate region. This scaling forms the basis for the achievable covert rates derived in Theorem~\ref{thm:main}.

\begin{lemma}[Asymptotic Scaling of Mutual Information] \label{lemma_D_theta_1}
In the covert regime where $N_S = s_n\in \omega(n^{-1/2})\cap o(1)$, the mutual information terms expand as follows:
\begin{align}
    I(\Theta_{X}\Theta_{Y}:BI_XI_Y)_{\sigma} &\geq \frac{-\kappa}{1+(1-\kappa)N_B}s_n\log s_n + O(s_n), \notag\\
    I(\Theta_{X}:\Theta_{Y}BI_XI_Y)_{\sigma} &\geq \frac{-\kappa\tau}{1+(1-\kappa)N_B}s_n\log s_n + O(s_n), \notag\\
    I(\Theta_{Y}:\Theta_{X}BI_XI_Y)_{\sigma} &\geq \frac{-\kappa(1-\tau)}{1+(1-\kappa)N_B}s_n\log s_n + O(s_n). \notag
\end{align}
Consequently, the achievable rate region converges to a rectangle defined by:
\begin{align}
    \left\{ (R_X, R_Y) \;\middle|\; 
    \begin{aligned}
        R_X &\leq \frac{-\kappa\tau}{1+(1-\kappa)N_B} s_n \log s_n + o(s_n\log s_n),  \\[1ex]
        R_Y &\leq \frac{-\kappa(1-\tau)}{1+(1-\kappa)N_B} s_n \log s_n + o(s_n\log s_n).
    \end{aligned} 
    \right\}.\label{eq:achieve-rate}
\end{align}
\end{lemma}
\begin{proof}
    Observe that the information quantities $I(\Theta_X : BI_X )_{\sigma}$, $I(\Theta_Y : BI_Y )_{\sigma'}$, $I(\Theta_X : BI_X )_{\sigma'}$, and $I(\Theta_Y : BI_Y )_{\sigma}$ correspond to effective point-to-point channel scenarios. Consequently, the proof proceeds analogously to \cite[Lemma~9]{Wang2024Resource}, which utilizes the results established in \cite[Theorems 2 and 5]{Grace2022Perturbation} and \cite[Appendix B]{Shi2020Practical}.
\end{proof}
\begin{remark}[Asymptotic Optimality and Geometry]\label{rem:optimality}
    In the low-photon regime ($\kappa s_n \ll N_B$), our achievable rate region collapses into a rectangle and asymptotically matches the fundamental point-to-point EA covert capacity derived in~\cite[Eq. 29]{Gagatsos2020Covert}. This optimality arises from two concurring effects: first, the vanishing signal-to-noise ratio renders multi-user interference negligible, effectively decoupling the channel; second, the information loss from our sub-optimal heterodyne measurement is asymptotically negligible compared to the dominant super-linear $O(s_n \log s_n)$ scaling of the covert mutual information.
\end{remark}

\section{Covert Communication over bosonic MAC}\label{sec:covert_MAC}
Building on the mutual information bounds established in the previous section, we now state our second main result regarding the covert achievability of the bosonic entanglement-assisted MAC. 

We first define formally the two-layered MAC code.
\begin{definition}[Two-layered MAC Code]
Let $n \in \mathbb{N}$ be the blocklength. Consider a sub-family of codes where the message sets for senders $A_X$ and $A_Y$ are decomposed as $M_{X_1}M_{X_2}$ and $M_{Y_1}M_{Y_2}$, respectively. An $(M_{X_1}, M_{Y_1}, M_{X_2}, M_{Y_2}, S_X, S_Y, l_x, l_y, \epsilon, \delta)$ two-layered MAC code consists of: 1. Encoding channels $\{\mathcal{E}_{M_{X_1}S_XM_{X_2} R_X^{l_x}\to A_X^n}^{(m_{X_1}, m_{X_2}, k_X)}\}$, $\{\mathcal{E}_{M_{Y_1}S_YM_{Y_2}R_Y^{l_y}\to A_Y^n}^{(m_{Y_1}, m_{Y_2}, k_Y)}\}$ that map message pairs, secret keys, and entanglement resources $l_x, l_y$ to channel input states;
2. A two-staged joint decoding POVM $(\{\Pi_{B^n}^{(m_{X_1}, m_{Y_1})}\}, \{\Gamma_{B^n I_X^{l_x}I_Y^{l_y}}^{(m_{X_1}, m_{Y_1}, m_{X_2}, m_{Y_2})}\})$ at the receiver (Bob);
is $\epsilon$-reliable and $\delta$-covert if the average decoding error over all messages and the joint covertness constraint satisfy:
\begin{align}
    &\mathbf{E}_{S_X,S_Y}\left[\mathbf{P}\left((\widehat{M}_{X_1}, \widehat{M}_{Y_1}, \widehat{M}_{X_2}, \widehat{M}_{Y_2}) \neq (M_{X_1}, M_{Y_1}, M_{X_2}, M_{Y_2})\right)\right] \le \epsilon, \label{eq:reliable_equa}\\
    & \frac{1}{2}  \left\| \hat{\sigma}_{W^n} - \sigma_{0,W}^{\otimes n} \right\|_1 \le \delta,\label{eq:covert_equa}
\end{align}
where $\hat{\sigma}_{W^n}$ is the average state induced at the warden (Willie) by both senders' transmissions and $\sigma_{0,W}^{\otimes n}$ denotes Willie's covert state corresponding to the absence of communication.
\end{definition}

The following theorem provides the main characterization of the two-layered MAC code. Note that the sequences $\alpha_n$ and $\beta_n$ govern the sparsity (transmission probability) of sender $A_X$ and sender $A_Y$, respectively, while $s_n$ quantifies the mean photon number of the active signal pulses.

\begin{theorem}\label{thm:main}
    Let $\varepsilon, \delta >0$. Let $\alpha_n, \beta_n$ and $s_n$ be sequence of positive numbers satisfying $\alpha_n, \beta_n, s_n \in \omega(n^{-\frac{1}{2}})\cap o(1)$. Besides, 
    \begin{align}
        (\alpha_n\tau + \beta_n(1-\tau))s_n\leq  \frac{2\sqrt{\kappa N_B(1+ \kappa N_B)}}{1-\kappa}Q^{-1}\left(\frac{1- \delta}{2}\right) n^{-\frac{1}{2}} - o(n^{-\frac{1}{2}}).\label{eq:covert_constraint}
    \end{align}
    Then, there exists $\bar\mu > 0$ and a sequence $(M_{X_1}, M_{Y_1}, M_{X_2}, M_{Y_2}, S_X, S_Y, l_x, l_y, \varepsilon, \delta)$ such that for $n$ large enough,
    \begin{align}
        &\log M_{X_1} \geq \frac{n\alpha_n\tau^2\kappa^2 s_n^2}{2(1-\kappa)N_B(1+ (1-\kappa)N_B)} + o(n\alpha_ns_n^2)\nonumber\\
        &\log M_{Y_1} \geq \frac{n\beta_n(1-\tau)^2\kappa^2 s_n^2}{2(1-\kappa)N_B(1+ (1-\kappa)N_B)} + o(n\beta_ns_n^2)\nonumber\\
        &\log M_{X_1} M_{Y_1} \geq \frac{n(\alpha_n\tau^2+\beta_n(1-\tau)^2)\kappa^2 s_n^2}{2(1-\kappa)N_B(1+ (1-\kappa)N_B)} + o(n(\alpha_n+ \beta_n)s_n^2)\nonumber\\
        &\log M_{X_1}S_X \leq \frac{n\alpha_n\tau^2(1-\kappa)^2 s_n^2}{2\kappa N_B(1+ \kappa N_B)} + o(n\alpha_ns_n^2)\nonumber\\
        &\log M_{Y_1}S_Y \leq \frac{n\beta_n(1-\tau)^2(1-\kappa)^2 s_n^2}{2\kappa N_B(1+ \kappa N_B)} + o(n\beta_ns_n^2)\nonumber\\
        &\log M_{X_1}S_XM_{Y_1}S_Y \leq \frac{n(\alpha_n\tau^2+\beta_n(1-\tau)^2)(1-\kappa)^2 s_n^2}{2\kappa N_B(1+ \kappa N_B)} + o(n(\alpha_n+ \beta_n)s_n^2)\nonumber\\
        &\log M_{X_2} \geq (1-\bar\mu)\frac{-n\kappa\tau\alpha_n}{1+(1-\kappa)N_B}s_n\log s_n + O(n \alpha_n s_n)\nonumber\\
        &  \log M_{Y_2} \geq (1-\bar\mu)\frac{-n\kappa (1-\tau)\beta_n}{1+(1-\kappa)N_B}s_n\log s_n + O(n \beta_n s_n)\nonumber\\
        &\log M_{X_2}M_{Y_2} \geq (1-\bar\mu)\frac{-n\kappa(\tau\alpha_n + (1-\tau)\beta_n)}{1+(1-\kappa)N_B}s_n\log s_n + O(n (\alpha_n + \beta_n)s_n),\nonumber
    \end{align}
    with $l_x = n\alpha_n(1-\bar\mu)$ TMSV pairs between $A_X$ and $B$,  and $l_y = n\beta_n(1-\bar\mu)$ TMSV pairs between $A_Y$ and $B$, is achievable using a codebook $(\mathcal{C}_{X_1}, \mathcal{C}_{Y_1}, \mathcal{C}_{X_2}, \mathcal{C}_{Y_2})$, where each TMSV pair corresponds to a TMSV state \eqref{eq:TMSV} with mean photon number $N_s = s_n$.
\end{theorem}
Theorem~\ref{thm:main} can be specialized to reveal the explicit constants governing the achievable rates when both senders share the same order of sparsity.

\begin{corollary}\label{cor:scaling}
For any $\gamma \in (0, 1/2)$, fix constants $\alpha, \beta, s > 0$ such that $(\alpha\tau + \beta(1-\tau))s = \frac{2\sqrt{\kappa N_B(1+ \kappa N_B)}}{1-\kappa}Q^{-1}\left(\frac{1- \delta}{2}\right)$. Let $\alpha_n = \alpha n^{-1/2+\gamma}$, $\beta_n = \beta n^{-1/2+\gamma}$, and $s_n = s n^{-\gamma}$. Then:
\begin{itemize}
    \item {Layer 1 Rates ($O(n^{1/2-\gamma})$ scaling):}
    \begin{align}
        \lim_{n\to\infty} \frac{\log M_{X_1}}{n^{1/2-\gamma}} &= \frac{\alpha\tau^2\kappa^2 s^2}{2(1-\kappa)N_B(1+ (1-\kappa)N_B)}\notag, \\
        \lim_{n\to\infty} \frac{\log M_{Y_1}}{n^{1/2-\gamma}} &= \frac{\beta(1-\tau)^2\kappa^2 s^2}{2(1-\kappa)N_B(1+ (1-\kappa)N_B)}.\notag
    \end{align}
    \item {Layer 2 Rates ($O(\sqrt{n}\log n)$ scaling):}
    \begin{align}
        \lim_{n\to\infty} \frac{\log M_{X_2}}{\sqrt{n}\log n} &= \frac{\gamma \kappa \tau \alpha s}{1+(1-\kappa)N_B}, \notag\\
        \lim_{n\to\infty} \frac{\log M_{Y_2}}{\sqrt{n}\log n} &= \frac{\gamma \kappa (1-\tau) \beta s}{1+(1-\kappa)N_B}.\notag
    \end{align}
    \item {Entangled Nats Consumption ($O(\sqrt{n}\log n)$ scaling):}
    \begin{align}
        \lim_{n\to\infty} \frac{l_xg(s_n)}{\sqrt{n}\log n} = \gamma\alpha s, \quad \lim_{n\to\infty} \frac{l_yg(s_n)}{\sqrt{n}\log n} = \gamma\beta s.\notag
    \end{align}
    where $g(x) = (x+1)\log(x+1) - x\log x$.
\end{itemize}
\end{corollary}
\begin{remark}[Linear Trade-off in Covert Resource Allocation]\label{rem:tradeoffs}
    Corollary~\ref{cor:scaling} shows that when sparsity parameters share the same asymptotic order ($\alpha, \beta > 0$), the covertness constraint imposes a strict linear trade-off, as the warden detects a weighted sum of individual signal powers. Consequently, any increase in Sender $A_X$'s activity $\alpha$ must be compensated by a proportional decrease in Sender $A_Y$'s $\beta$, requiring users to coordinate their transmission frequencies to ensure the aggregate energy remains within the global covert budget.
\end{remark}
\begin{remark}[The Case with a More Active Sender]\label{rem:dominance}
If the sparsity sequences $\alpha_n$ and $\beta_n$ are of different orders, the covertness constraint in \eqref{eq:covert_constraint} is dominated by the sequence with the higher order. This implies that the user with the larger $\alpha_n$ or $\beta_n$ dictates the maximum allowable signal power $s_n$ for both users to ensure joint covertness.
\end{remark}
\begin{remark}[Functional Separation of Encoding Layers]\label{comment}
    In the large-$n$ limit, the Layer 2 rate dominates the total throughput, achieving $O(\sqrt{n} \log n)$ scaling due to entanglement assistance. Layer 1 serves primarily as a covertness regulator, using sparsity to mask the transmission as background noise.
\end{remark}

We outline the coding scheme and proof strategy for Theorem~\ref{thm:main}, with the comprehensive proof provided in Section~\ref{sec:covert_proof}.

\textbf{1. Two-Layer Codebook Construction.}
The scheme employs a nested codebook structure for both senders. \textbf{Layer 1} uses sparse encoding combined with shared secret keys to determine pulse positions, ensuring covertness against Willie. \textbf{Layer 2} utilizes high-order PSK modulation applied to shared TMSV idlers. The resulting modulated quantum states are embedded into the temporal positions selected by Layer 1.

\textbf{2. Decoding Strategy.}
Bob performs sequential decoding: he first utilizes the secret keys to resolve the Layer 1 pulse positions. Upon success, he applies a phase-decoding procedure to the idler systems at those specific positions to recover the Layer 2 messages.

\textbf{3. Reliability and Covertness.}
Reliability is established by applying Lemma~\ref{lemm:reliable_resolvable} across both layers. For Layer 1, Bob first projects the infinite-dimensional bosonic system onto a suitable finite-dimensional subspace to enable joint decoding. Covertness is proven in three steps: (i) using Lemma~\ref{lemm:reliable_resolvable} to ensure resolvability at Willie’s receiver; (ii) showing the average transmitted state after sparse coding approximates an i.i.d. mixture (Lemma~\ref{lemma:typical_product}); and (iii) proving this i.i.d. state is statistically indistinguishable from the vacuum state (Lemma~\ref{lemma:covertstate}).

\section{Proof of Theorem~\ref{thm:main}}\label{sec:covert_proof}
\subsection{Code Construction}
Let $\mathcal{X}= \{x_0, x_1\}$ and $\mathcal{Y}= \{y_0, y_1\}$. We define the binary distributions $P_X(x) \triangleq (1-\alpha_n)\mathbf{1}_{x_0} + \alpha_n \mathbf{1}_{x_1}$ and $P_Y(y) \triangleq (1-\beta_n)\mathbf{1}_{y_0} + \beta_n \mathbf{1}_{y_1}$.
We further define the truncated $n$-fold distributions, which post-select for symbols satisfying the minimal weight threshold:
\begin{align}
    &\tilde{P}_{X^n}(x^n) \triangleq \frac{P^{\otimes n}_X(x)\mathbf{1}\{\hat{p}_{X}(x_1) \geq (1-\bar{\mu})\alpha_n)\}}{\mathbf{Pr}\{\hat{p}_{X}(x_1) \geq (1-\bar{\mu})\alpha_n)\}},\\
    &\tilde{P}_{Y^n}(y^n) \triangleq \frac{P^{\otimes n}_Y(y)\mathbf{1}\{\hat{p}_{Y}(y_1) \geq (1-\bar{\mu})\beta_n)\}}{\mathbf{Pr}\{\hat{p}_{Y}(y_1) \geq (1-\bar{\mu})\beta_n)\}}.
\end{align}
Lemma~\ref{lemma:typical_product} will later establish that these truncated distributions are trace-distance close to their i.i.d. counterparts.

\textbf{Codebook Generation.}
Senders $A_X$ and $A_Y$ share secret keys $k_X \in [S_X]$, $k_Y \in [S_Y]$ and entangled TMSV pairs $\ket{\psi}_{RI} = \sum_{k}\sqrt{s_n^k(s_n+1)^{-(k+1)}}\ket{k}_R\ket{k}_I$ with Bob. 
\begin{enumerate}
    \item \textbf{Layer 1 (Sparsity):} $A_X$ (resp. $A_Y$) generates $M_{X_1} S_X$ (resp. $M_{Y_1} S_Y$) codewords $x_{m_{X_1} k_X}$ (resp. $y_{m_{Y_1} k_Y}$) i.i.d. from $\tilde{P}_{X^n}$ (resp. $\tilde{P}_{Y^n}$) to form codebooks $\mathcal{C}_{X_1}$ and $\mathcal{C}_{Y_1}$. By definition, these codewords contain at least $l_x = n(1 - \bar{\mu}) \alpha_n$ (resp. $l_y = n(1 - \bar{\mu}) \beta_n$) active pulses.
    \item \textbf{Layer 2 (Modulation):} $A_X$ (resp. $A_Y$) generates $M_{X_2}$ (resp. $M_{Y_2}$) codewords $\theta_{m_{X_2}}$ (resp. $\phi_{m_{Y_2}}$) uniformly from the product constellation $\Theta_{X}^{\otimes l_x}$ (resp. $\Theta_{Y}^{\otimes l_y}$), where $\Theta_X$ and $\Theta_Y$ denotes a $2^n$-PSK constellation.
\end{enumerate}

\textbf{Encoding.}
To transmit message tuple $(m_{X_1}, k_X, m_{X_2})$, $A_X$ first selects the Layer 1 sequence $x_{m_{X_1} k_X}$. She then modulates her share of the entanglement $R_X^{l_x}$ using the Layer 2 codeword $\theta_{m_{X_2}}$ via the unitary phase rotation operator $\hat{U}_{\theta} =\exp\left\{j2\theta\hat{a}^\dagger \hat{a}\right\}$, yielding:
\begin{align}
    \psi_{R_X^{l_x}I_X^{lx}}^{\theta_{m_{X_2}}} \triangleq \bigotimes_{i = 1}^{l_x} \left(\left( \hat{U}_{(\theta_{m_{X_2}})_i}\otimes \mathrm{id}_{I_X}\right)\ket{\psi}\bra{\psi}_{R_iI_i}\left( \hat{U}^{\dagger}_{(\theta_{m_{X_2}})_i}\otimes \mathrm{id}_{I_X}\right)\right).
\end{align}
These modulated states are embedded into the $n$-mode channel at positions corresponding to the $1$s in $x_{m_{X_1}k_X}$. The remaining positions are filled with thermal noise $\rho_{s_n}$, resulting in the transmitted state:
\begin{align}
    \sigma_{A^n_XI_X^{l_x}}^{x_{m_{X_1}k_X}, \theta_{m_{X_2}}} \triangleq &\psi^{(\theta_{m_{X_2}})_1}_{(A_X)_1I_1}\otimes \ket{0}\bra{0}_{(A_X)_2}\otimes \psi^{(\theta_{m_{X_2}})_2}_{(A_X)_3I_3}\otimes \cdots\notag\\
    &\cdots\otimes (\rho_{s_n})_{(A_X)_{n-1}}\otimes \ket{0}\bra{0}_{(A_X)_n}.
\end{align}
Sender $A_Y$ performs an analogous procedure to generate $\sigma_{A^n_YI_Y^{l_y}}^{y_{m_{Y_1}k_Y}, \phi_{m_{Y_2}}}$.

\textbf{Transmission States.}
The encoding process induces the following classical-quantum (c-q) states for the input systems:
\begin{align}
    &\sigma_{X^n\Theta_{X}^{l_x}A^n_XI_X^{l_x}}(m_{X_1}, m_{X_2}, k_X) \triangleq \ket{x_{m_{X_1}k_X}}\bra{x_{m_{X_1}k_X}} \otimes \ket{\theta_{m_{X_2}}}\bra{\theta_{m_{X_2}}}\otimes \sigma_{A^n_XI_X^{l_x}}^{x_{m_{X_1}k_X}, \theta_{m_{X_2}}};\notag\\
    & \sigma_{Y^n\Theta_{Y}^{l_y}A^n_YI_Y^{l_y}}(m_{Y_1}, m_{Y_2}, k_Y) \triangleq\ket{y_{m_{Y_1}k_Y}}\bra{y_{m_{Y_1}k_Y}} \otimes \ket{\phi_{m_{Y_2}}}\bra{\phi_{m_{Y_2}}}\otimes \sigma_{A^n_YI_Y^{l_y}}^{y_{m_{Y_1}k_Y}, \phi_{m_{Y_2}}}.\notag
\end{align}
After transmission through the channel $(\mathcal{L}_{A_XA_Y\to BW}^{(\tau, \kappa, N_B)})^{\otimes n}$, the joint state of Bob and Willie is:
\begin{align}
    &\sigma_{X^nY^n\Theta_{X}^{l_x}\Theta_{Y}^{l_y}B^nI^{l_x}_X I_Y^{l_y}W^n}(m_{X_1}, m_{X_2}, k_X, m_{Y_1}, m_{Y_2}, k_Y)\notag\\
    & \triangleq \left(\left(\mathcal{L}_{A_XA_Y\to BW}^{(\tau, \kappa, N_B)}\right)^{\otimes n}\otimes \mathrm{id}_{I_X^{l_x}I_Y^{l_y}}\right)\notag\\
    &\quad \quad \bigg(\sigma_{X^n\Theta_{X}^{l_x}A^n_XI_X^{l_x}}(m_{X_1}, m_{X_2}, k_X)\otimes \sigma_{Y^n\Theta_{Y}^{l_y}A^n_YI_Y^{l_y}}(m_{Y_1}, m_{Y_2}, k_Y)\bigg).\label{eq:c-qstate}
\end{align}
Taking the expectation over all codebooks and randomization variables, the average joint c-q state is:
\begin{align}
    &\tilde\sigma_{X^nY^n\Theta_{X}^{l_x}\Theta_{Y}^{l_y}B^nI^{l_x}_X I_Y^{l_y}W^n} \notag\\
    &= \sum_{x^n,y^n}\tilde{P}_{X^n}(x^n)\tilde{P}_{Y^n}(y^n)\ket{x^n}\bra{x^n}\otimes \ket{y^n}\bra{y^n} \notag\\
    &\quad\quad \otimes \frac{1}{|\Theta_{X}^{l_x}||\Theta_{Y}^{l_y}|}\sum_{\theta^{l_x}\in \Theta_{X}^{\otimes l_x}} \sum_{\phi^{l_y}\in \Theta_{Y}^{\otimes l_y}} \ket{\theta^{l_x}}\bra{\theta^{l_x}}\otimes \ket{\phi^{l_y}}\bra{\phi^{l_y}} \otimes \sigma^{x^n, y^n, \theta^{l_x}, \phi^{l_y}}_{B^nW^nI_X^{l_x}I_Y^{l_y}}.\label{eq:jointstate}
\end{align}
Since Willie lacks access to the idler systems $I_X I_Y$, his state decouples from the phase modulation $\Theta$ and Layer 2 messages. His effective observation for a fixed codebook tuple $(\mathcal{C}_{X_1},\mathcal{C}_{Y_1})$ is the average over all keys and messages:
\begin{align}
    \hat{\sigma}_{W^n} = \frac{1}{|M_{X_1}||M_{Y_1}||S_X||S_Y|} \sum_{m_{X_1}, m_{Y_1}}  \sum_{k_X, k_Y} \sigma_{W^n}(m_{X_1}, k_X, m_{Y_1}, k_Y).
\end{align}
In the absence of transmission, Willie observes the covert state $\sigma_{0,W}^{\otimes n} = \rho_{\kappa N_B}^{\otimes n}$.

\subsection{Lemmas}
Recall that Lemma~\ref{lemm:reliable_resolvable} establishes the fundamental one-shot reliability and resolvability conditions required for our codebook existence argument. 
Here, we develop additional lemmas that will be applied in the proof of Theorem~\ref{thm:main}. 

The following classical-quantum states are defined specifically for Lemma~\ref{lemma:typical_product}
\begin{align}
    &\sigma_{XYA} \triangleq \sum_{x\in \mathcal{X}} P_X(x)\ket{x}\bra{x}_X\otimes \sum_{y\in \mathcal{Y}} P_Y(y)\ket{y}\bra{y}_Y \otimes \sigma^{x,y}_A. \label{sigma_XYA}\\
    &\tilde{\sigma}_{X^nY^nA^n} \triangleq \sum_{x^n, y^n}\tilde{P}_{X^n}(x^n)\tilde{P}_{Y^n}(y^n)\ket{x^n}\bra{x^n}_{X^n}\otimes\ket{y^n}\bra{y^n}_{Y^n}\otimes \sigma^{x^n,y^n}_{A^n}. \label{tilde_sigma_XYA}
\end{align}
Lemma~\ref{lemma:typical_product} below characterizes the closeness between the original i.i.d. state $\sigma_{XYA}^{\otimes n}$ and the state $\tilde{\sigma}_{X^nY^nA^n}$ after discarding the codewords with small weights.
\begin{lemma}\label{lemma:typical_product}
For $\sigma_{XYA}$, and $\tilde{\sigma}_{X^nY^nA^n}$ defined in \eqref{sigma_XYA} and \eqref{tilde_sigma_XYA}, and for $n$ large enough,
\begin{align}
    \frac{1}{2}\left\|\tilde\sigma_{A^n}- \sigma_A^{\otimes n}\right\|_1 &\leq \max\left\{\frac{1}{2}\left\|\tilde\sigma_{X^nA^n}- \sigma_{XA}^{\otimes n}\right\|_1, \frac{1}{2}\left\|\tilde\sigma_{Y^nA^n}- \sigma_{YA}^{\otimes n}\right\|_1\right\}\notag \\
    & \leq \frac{1}{2}\left\|\tilde\sigma_{X^nY^nA^n}- \sigma_{XYA}^{\otimes n}\right\|_1\notag\\
    &= \frac{1}{2}\left\|\tilde\sigma_{X^nY^n}- \sigma_{XY}^{\otimes n}\right\|_1\notag\\
    & \leq 2\mathbf{e}^{-\frac{1}{2}\bar\mu^2n\alpha_n} + 2\mathbf{e}^{-\frac{1}{2}\bar\mu^2n\beta_n},\notag
\end{align}
and
\begin{align}
    &\frac{1}{2}\left\|\tilde\sigma_{X^nY^n}\otimes\tilde\sigma_{A^n}- \sigma_{XY}^{\otimes n}\otimes \sigma_{A}^{\otimes n}\right\|_1 \leq 4\mathbf{e}^{-\frac{1}{2}\bar\mu^2n\alpha_n} + 4\mathbf{e}^{-\frac{1}{2}\bar\mu^2n\beta_n}\notag\\
    &     \frac{1}{2}\left\|\tilde\sigma_{X^n}\otimes\tilde\sigma_{Y^nA^n}- \sigma_{X}^{\otimes n}\otimes \sigma_{YA}^{\otimes n}\right\|_1 \leq 4\mathbf{e}^{-\frac{1}{2}\bar\mu^2n\alpha_n} + 4\mathbf{e}^{-\frac{1}{2}\bar\mu^2n\beta_n}\notag\\
    &     \frac{1}{2}\left\|\tilde\sigma_{Y^n}\otimes\tilde\sigma_{X^nA^n}- \sigma_{Y}^{\otimes n}\otimes \sigma_{XA}^{\otimes n}\right\|_1 \leq 4\mathbf{e}^{-\frac{1}{2}\bar\mu^2n\alpha_n} + 4\mathbf{e}^{-\frac{1}{2}\bar\mu^2n\beta_n}.\notag
\end{align}
\end{lemma}
\begin{proof}
    The result follows from the usage of Chernoff bound, triangle inequality, and the monotonicity of the trace distance.
\end{proof}
In the next three lemmas, we consider the following one-shot state
\begin{align}
    &\sigma_{XYA_XA_Y} \triangleq \sum_{x\in \mathcal{X}} P_X(x)\ket{x}\bra{x}_X\otimes \sum_{y\in \mathcal{Y}} P_Y(y)\ket{y}\bra{y}_Y \otimes \sigma^{x}_{A_X}\otimes \sigma^{y}_{A_Y}, \label{sigma_XYAXAY}
\end{align}
where $\sigma_{A_X}^x = \ket{0}\bra{0}\mathbf{1}\{x=x_0\} + \rho_{s_n}\mathbf{1}\{x=x_1\}$, and $\sigma_{A_Y}^y = \ket{0}\bra{0}\mathbf{1}\{y=y_0\} + \rho_{s_n}\mathbf{1}\{y=y_1\}$.

Lemma~\ref{lemma:covertstate} below shows that the product of the sparsity parameters $\alpha_n, \beta_n$ and the signal mean photon number $s_n$ controls the covertness.
\begin{lemma}\label{lemma:covertstate}
    Let $\sigma_{XYBW} \triangleq \mathcal{L}_{A_XA_Y\to BW}^{(\tau,\kappa, N_B)}(\sigma_{XYA_XA_Y})$, where $\sigma_{XYA_XA_Y}$ is defined in \eqref{sigma_XYAXAY}. The trace distance between $\sigma_W^{\otimes n}$ and $\sigma_{0,W}^{\otimes n}= \rho_{\kappa N_B}^{\otimes n}$ is 
    \begin{align}
        &\frac{1}{2} \left\|\sigma_W^{\otimes n}-\sigma_{0,W}^{\otimes n}\right\|_1 \notag\\
        &\leq 1 - 2\mathrm{Q}\left(\frac{\sqrt{n}(\alpha_n\tau + \beta_n(1-\tau))(1-\kappa)s_n}{2\sqrt{\kappa N_B(1+\kappa N_B)}}\right) + \frac{D_1}{\sqrt{n}} \notag\\
        &\quad + D_0\sqrt{n}((\alpha_n+\beta_n)s_n^3+(\alpha_n+\beta_n)^2s_n).\notag
    \end{align}
\end{lemma}
\begin{proof}
    The proof follows analogously from the derivation of \cite[Lemma IV.1]{Wang2022Towards}.
\end{proof}

The following lemma is applied in the throughput analysis of the Layer 1. For clarity of presentation, the corresponding lemma for Layer 2 (e.g. Lemma~\ref{lemma_D_theta}) will be stated later in Section~\ref{sec:throughput_layer_2}.
\begin{lemma}\label{lemma:throughput}
    Let $\sigma_{XYBW} \triangleq \mathcal{L}_{A_XA_Y\to BW}^{(\tau,\kappa, N_B)}(\sigma_{XYA_XA_Y})$, where $\sigma_{XYA_XA_Y}$ is defined in \eqref{sigma_XYAXAY}.  Then, 
    \begin{align}
        &D(\sigma_{XYB}\|\sigma_{XY}\otimes \sigma_B) = \frac{(\alpha_n\tau^2+\beta_n(1-\tau)^2)\kappa^2 s_n^2}{2(1-\kappa)N_B(1+ (1-\kappa)N_B)} \notag\\
        &\quad\quad\quad\quad\quad\quad\quad\quad\quad\quad + O((\alpha_n+ \beta_n)^2s_n^2, (\alpha_n+\beta_n)s_n^3)\notag\\
        &V(\sigma_{XYB}\|\sigma_{XY}\otimes \sigma_B) = \frac{(\alpha_n\tau^2+\beta_n(1-\tau)^2)\kappa^2 s_n^2}{1-\kappa)N_B(1+ (1-\kappa)N_B)} \notag\\
        &\quad\quad\quad\quad\quad\quad\quad\quad\quad\quad+  O((\alpha_n+ \beta_n)^2s_n^2, (\alpha_n+\beta_n)s_n^3)\notag\\
        &R(\sigma_{XYB}\|\sigma_{XY}\otimes \sigma_B) = O((\alpha_n+\beta_n)s_n^4)\notag\\
        &D(\sigma_{XYB}\|\sigma_{X}\otimes \sigma_{YB}) = \frac{\alpha_n\tau^2\kappa^2 s_n^2}{2(1-\kappa)N_B(1+ (1-\kappa)N_B)} + O(\alpha_n^2s_n^2, \alpha_n\beta_ns_n^2)\notag\\
        &V(\sigma_{XYB}\|\sigma_{X}\otimes \sigma_{YB}) = \frac{\alpha_n\tau^2\kappa^2 s_n^2}{(1-\kappa)N_B(1+ (1-\kappa)N_B)} + O(\alpha_n^2s_n^2, \alpha_n\beta_ns_n^2)\notag\\
        &R(\sigma_{XYB}\|\sigma_{X}\otimes \sigma_{YB}) = O(\alpha_n s_n^4)\notag
    \end{align}
    \begin{align}
        &D(\sigma_{XYB}\|\sigma_{Y}\otimes \sigma_{XB}) = \frac{\beta_n(1-\tau)^2\kappa^2 s_n^2}{2(1-\kappa)N_B(1+ (1-\kappa)N_B)} + O(\beta_n^2s_n^2, \alpha_n\beta_ns_n^2)\notag\\
        &V(\sigma_{XYB}\|\sigma_{Y}\otimes \sigma_{XB}) = \frac{\beta_n(1-\tau)^2\kappa^2 s_n^2}{(1-\kappa)N_B(1+ (1-\kappa)N_B)} + O(\beta_n^2s_n^2, \alpha_n\beta_ns_n^2)\notag\\
        &R(\sigma_{XYB}\|\sigma_{Y}\otimes \sigma_{XB}) = O(\beta_n s_n^4)\notag
    \end{align}
The corresponding entropic quantities of $\sigma_{XYW}$ take the same form, with the substitution being $\kappa \leftrightarrow 1-\kappa$.
\end{lemma}
\begin{proof}
    The proof follows from a direct expansion of the information quantities in terms of $\alpha_n$, $\beta_n$, and $s_n$.
\end{proof}

In Layer~$1$, to apply Lemma~\ref{lemm:reliable_resolvable}, Bob must first implement the projection channel that maps the state $\tilde\sigma_{X^nY^nB^nW^n}$ to $\tilde\sigma_{X^nY^nZ^nW^n}$, in order to have his local system reside in finite dimension. When the trace distance between the relevant states is sufficiently small, the following lemma provides a second-order expansion of the entropic quantities that will be useful in our analysis.
\begin{lemma}\label{lemm:expansion}
    Let $\sigma_{XYBW} \triangleq \mathcal{L}_{A_XA_Y\to BW}^{(\tau,\kappa, N_B)}(\sigma_{XYA_XA_Y})$, where $\sigma_{XYA_XA_Y}$ is defined in \eqref{sigma_XYAXAY}. Assume that the states $\tilde\sigma_{X^nY^nB^nW^n}$ and $\tilde\sigma_{X^nY^nZ^nW^n}$ defined in \eqref{eq:jointstate} satisfy
    \begin{align}
        &\frac{1}{2}\left\|\tilde\sigma_{X^nY^nZ^n} - \tilde\sigma_{X^nY^nB^n}\right\|_1 \leq \delta_n \leq \exp(-\omega(\sqrt{n}));\notag\\
        &\frac{1}{2}\left\|\tilde\sigma_{X^nY^nB^n} - \sigma_{XYB}^{\otimes n}\right\|_1 \leq 2\exp\left(-\frac{1}{2}\bar\mu^2n\alpha_n\right) + 2\exp\left(-\frac{1}{2}\bar\mu^2n\beta_n\right).\notag
    \end{align}
    Then, the following holds.
    \begin{align}
        &D_H^{\varepsilon^2-n^{-\frac{1}{2}}}(\tilde{\sigma}_{X^nY^nZ^n}\|\tilde{\sigma}_{X^nY^n}\otimes \tilde{\sigma}_{Z^n})\notag\\
        &\geq nD(\sigma_{XYB}\|\sigma_{XY}\otimes \sigma_{B}) - \sqrt{nV(\sigma_{XYB}\|\sigma_{XY}\otimes \sigma_{B})}Q^{-1}(\varepsilon^2)\notag\\
        &\quad+ O\left(\frac{R(\sigma_{XYB}\|\sigma_{XY}\otimes \sigma_{B})^\frac{3}{4}}{V(\sigma_{XYB}\|\sigma_{XY}\otimes \sigma_{B})}\right).\notag
    \end{align}
    Besides, if the distance between $\tilde\sigma_{X^nY^nW^n}$ and $\sigma_{XYW}^{\otimes n}$ satisfies
    \begin{align}
        \frac{1}{2}\left\|\tilde\sigma_{X^nY^nW^n} - \sigma_{XYW}^{\otimes n}\right\|_1 \leq 2\exp\left(-\frac{1}{2}\bar\mu^2n\alpha_n\right) + 2\exp\left(-\frac{1}{2}\bar\mu^2n\beta_n\right),\notag
    \end{align}
    the following holds.
    \begin{align}
        &D_{\max}^{\delta- n^{-\frac{1}{2}}}(\tilde\sigma_{X^nY^nW^n}\|\tilde\sigma_{X^nY^n}\otimes \tilde\sigma_{W^n}) \notag\\
        &\leq nD(\sigma_{XYW}\|\sigma_{XY}\otimes \sigma_{W}) + \sqrt{nV(\sigma_{XYW}\|\sigma_{XY}\otimes \sigma_{W})}Q^{-1}(\delta^2)\notag\\
        &\quad + O\left(\frac{R(\sigma_{XYW}\|\sigma_{XY}\otimes \sigma_{W})^\frac{3}{4}}{V(\sigma_{XYW}\|\sigma_{XY}\otimes \sigma_{W})}\right).\notag
    \end{align}
\end{lemma}
\begin{proof}
    The lower bound expansion for the hypothesis testing relative entropy $D_H^{\varepsilon^2-n^{-1/2}}$ follows directly from \cite[Proposition~13]{Oskouei:2018oeu}. 
    For the upper bound on the smooth max-relative entropy $D_{\max}^{\delta- n^{-1/2}}$, we first invoke \cite[Lemma~10]{Khatri2021Second} to relate it to the hypothesis testing quantity $D_H^{1-\delta^2 - O(n^{-1/2})}$. The final asymptotic expansion is then established by applying Proposition~2 of \cite[Appendix C]{Kaur2017Upper}.
\end{proof}

\subsection{Bob's Projection Channel}\label{sec:projection_channel}
Recall the joint c-q state prior to measurement given in \eqref{eq:jointstate}.
To apply the finite-dimensional coding results of Lemma~\ref{lemm:reliable_resolvable}, Bob applies a truncation channel $\mathcal{N}_{B^nI_X^{l_x}I_Y^{l_y}\to Z^nJ_X^{l_x}J_Y^{l_x}}$ given in \eqref{eq:projection_channel_1}, which maps his infinite-dimensional bosonic systems ($B^n, I_X^{l_x}, I_Y^{l_y}$) to finite dimensional systems ($Z^n, J_X^{l_x}, J_Y^{l_y}$). 
Implementing this channel yields the effective output classical-quantum state 
\begin{align}
    \tilde\sigma_{X^nY^n\Theta_{X}^{l_x}\Theta_{Y}^{l_y}Z^nJ^{l_x}_X J_Y^{l_y}W^n} = \mathcal{N}_{B^nI_X^{l_x}I_Y^{l_y}\to Z^nJ_X^{l_x}J_Y^{l_x}} \left(\tilde\sigma_{X^nY^n\Theta_{X}^{l_x}\Theta_{Y}^{l_y}B^nI^{l_x}_X I_Y^{l_y}W^n}\right)\label{eq:projected_state}.
\end{align}

We select the projections $P_{n,B}, P_{n,I_X}, P_{n,I_Y}$ defining the channel $\mathcal{N}$ with sufficiently large rank such that the truncation error is uniformly bounded for all typical codewords:
\begin{align}
    \sup_{(x^n,y^n)\in \textrm{supp}(\tilde{P}_{X^n}\times \tilde{P}_{Y^n})}\frac{1}{2}\left\|\sigma_{\Theta_{X}^{l_x}\Theta_{Y}^{l_y}B^nI^{l_x}_X I_Y^{l_y}W^n}^{x^n,y^n} - \sigma_{\Theta_{X}^{l_x}\Theta_{Y}^{l_y}Z^nJ^{l_x}_X J_Y^{l_y}W^n}^{x^n,y^n}\right\|_1 \leq \delta_n,\label{eq:projection}
\end{align}
where $\delta_n$ is a sequence chosen small enough to satisfy Lemma~\ref{lemm:expansion}. This bound has two immediate consequences:

\begin{enumerate}
    \item \textbf{Layer 1 Consistency:} Averaging \eqref{eq:projection} over $x^n, y^n$ implies that the distance between the expected states is bounded:
    \begin{align}
         \frac{1}{2}\left\|\tilde\sigma_{X^nY^nB^n} - \tilde\sigma_{X^nY^nZ^n}\right\|_1 \leq \delta_n. \label{eq:projection_layer_1}
    \end{align}
    \item \textbf{Layer 2 Consistency:} Conditioned on successful Layer 1 decoding of $(x^n, y^n)$, Bob focuses on the subset of systems correlated with the idlers. By the monotonicity of trace distance under partial trace, \eqref{eq:projection} ensures that for the relevant subsystems $B^l I_X^{l_x} I_Y^{l_y}$:
    \begin{align}
      \frac{1}{2}\left\|\sigma_{\Theta_{X}^{l_x}\Theta_{Y}^{l_y}B^lI^{l_x}_X I_Y^{l_y}}^{x^n,y^n} - \sigma_{\Theta_{X}^{l_x}\Theta_{Y}^{l_y}Z^lJ^{l_x}_X J_Y^{l_y}}^{x^n,y^n}\right\|_1 \leq \delta_n. \label{eq:projection_layer_2}
    \end{align}
\end{enumerate}
Equations \eqref{eq:projection_layer_1} and \eqref{eq:projection_layer_2} allow us to substitute the infinite-dimensional states with their finite approximations in the subsequent reliability analysis.

\subsection{Reliability for Layer 1}
We analyze Layer 1 by considering the reduced state $\tilde{\sigma}_{X^n Y^n Z^n W^n}$, obtained by tracing out the entangled idler subsystems from the projected state \eqref{eq:projected_state}. Applying the reliability bound from Lemma~\ref{lemm:reliable_resolvable} with $\varepsilon = \varepsilon_1$, we limit the average decoding error probability as follows:
\begin{align}
    \mathbf{E}_{\mathcal{C}_{X_1}, \mathcal{C}_{Y_1}, S_X, S_Y} \left[\mathbf{P}\left((\hat{M}_{X_1}, \hat{M}_{Y_1}) \neq (M_{X_1}, M_{Y_1})\mid S_X, S_Y\right)\right] \leq 46\varepsilon_1.
\end{align}

\subsection{Covertness Analysis}\label{sec:covert}
We establish covertness by bounding the total trace distance between Willie's observation and the vacuum state using a three-step triangle inequality. First, we ensure the ideal i.i.d. output state $\sigma_W^{\otimes n}$ is close to the vacuum $\sigma_{0,W}^{\otimes n}$. By applying Lemma~\ref{lemma:covertstate} and enforcing the power constraint:
\begin{align}
    (\alpha_n\tau + \beta_n(1-\tau))s_n &= \frac{2\sqrt{\kappa N_B(1+ \kappa N_B)}}{1-\kappa}Q^{-1}\left(\frac{1- (\delta-\eta)}{2}\right)n^{-\frac{1}{2}} \notag\\
    &\quad - \frac{\bar{D}(s_n^2+\alpha_n+\beta_n)}{\sqrt{n}}- O(n^{-1}),\label{eq:covert_mid_constraint}
\end{align}
we obtain the bound for sufficiently large $\bar{D}$:
\begin{align}
    \frac{1}{2}\left\|\sigma_{W}^{\otimes n} - \sigma_{0,W}^{\otimes n}\right\|_1 \leq \delta - \eta.\label{eq:covert1}
\end{align}

Second, Lemma~\ref{lemma:typical_product} guarantees that the truncated average state $\tilde{\sigma}_{W^n}$ sufficiently approximates the i.i.d. state, for any $\eta' > 0$ and sufficiently large $n$:
\begin{align}
    \frac{1}{2}\left\|\tilde{\sigma}_{W^n}- \sigma_{W}^{\otimes n}\right\|_1 \leq 2\mathbf{e}^{-\frac{1}{2}\bar\mu^2n\alpha_n} + 2\mathbf{e}^{-\frac{1}{2}\bar\mu^2n\beta_n} \leq \eta'.\label{eq:covert2} 
\end{align}

Third, we apply the resolvability bound \eqref{eq:covert_in_lemma} from Lemma~\ref{lemm:reliable_resolvable} to the marginal state $\tilde\sigma_{X^nY^nZ^nW^n}$ of \eqref{eq:projected_state}. By choosing $\delta'' \ll \eta - \eta'$, we ensure that for the majority of codebooks, the induced state $\hat{\sigma}_{W^n}$ is close to the truncated average:
\begin{align}
    \mathbf{E}_{\mathcal{C}_{X_1}, \mathcal{C}_{Y_1}}\left[\frac{1}{2}\left\|\hat{\sigma}_{W^n}- \tilde{\sigma}_{W^n}\right\|_1\right] \leq \delta''.\label{eq:covert3}
\end{align}

Combining \eqref{eq:covert1}, \eqref{eq:covert2}, and 
\eqref{eq:covert3} via the triangle inequality and applying Markov's inequality, we conclude that a randomly drawn codebook tuple $(\mathcal{C}_{X_1}, \mathcal{C}_{Y_1})$ satisfies the covertness requirement below with high probability:
\begin{align}
    \frac{1}{2}\left\|\hat{\sigma}_{W^n} - \sigma_{0,W}^{\otimes n}\right\|_1 \leq \delta. \label{eq:covertness}
\end{align}

\subsection{Reliability for Layer 2 and Total Error Analysis}
From the Layer 1 analysis (Lemma~\ref{lemm:reliable_resolvable}), the average success probability over random codebooks and keys satisfies:
\begin{align}
    &\mathbf{E}_{\mathcal{C}_{X_1}, \mathcal{C}_{Y_1}, S_X, S_Y} \left[ \frac{1}{M_1K} \sum_{k, m_1} \mathrm{Tr}\left(\Lambda^{x_{m_{X_1}k_X}, y_{m_{Y_1}k_Y}}_{Z^n} \sigma_{Z^n}^{x_{m_{X_1}k_X}, y_{m_{Y_1}k_Y}}\right) \right] \notag\\
    &\geq 1-46\varepsilon_1. \label{eq:avg_success_prob}
\end{align}
Consider a specific codebook realization $(\mathcal{C}_{X_1}, \mathcal{C}_{Y_1})$ and keys $(k_X, k_Y)$ where the conditional success probability for a message pair $(m_{X_1}, m_{Y_1})$ is bounded by:
\begin{align}
    \mathrm{Tr}\left[\Lambda^{x_{m_{X_1}k_X}, y_{m_{Y_1}k_Y}}_{Z^n} \sigma_{Z^n}^{x_{m_{X_1}k_X}, y_{m_{Y_1}k_Y}}\right] \geq 1-46\varepsilon'_1. \label{eq:success_prob}
\end{align}
This implies that for any Layer 2 codewords $\theta^{l_x} \in \Theta_{X}^{\otimes l_x}$ and $\phi^{l_y} \in \Theta_{Y}^{\otimes l_y}$, the joint success probability including the idler systems is similarly bounded:
\begin{align}
    \mathrm{Tr}\left[\left(\Lambda^{x_{m_{X_1}k_X}, y_{m_{Y_1}k_Y}}_{Z^n}\otimes \mathrm{id}_{J_XJ_Y}\right) \sigma_{Z^nJ_X^{l_x}J_Y^{l_y}}^{x_{m_{X_1}k_X}, y_{m_{Y_1}k_Y}, \theta^{l_x}, \phi^{l_y}}\right] \geq 1-46\varepsilon'_1. \label{eq:general_success_prob}
\end{align}
Note that by \eqref{eq:avg_success_prob} and \eqref{eq:success_prob}, the average error parameter satisfies $\mathbf{E}[\varepsilon'_1] \leq \varepsilon_1$.

Let $\sigma'$ denote the post-measurement state conditioned on successful Layer 1 decoding:
\begin{align}
    (\sigma')^{x_{m_{X_1}k_X}, y_{m_{Y_1}k_Y}, \theta^{l_x}, \phi^{l_y}}_{Z^nJ_X^{l_x}J_Y^{l_y}} 
    = \frac{\sqrt{\Lambda} \, \sigma_{Z^nJ_X^{l_x}J_Y^{l_y}}^{x_{m_{X_1}k_X}, y_{m_{Y_1}k_Y}, \theta^{l_x}, \phi^{l_y}} \, \sqrt{\Lambda}}{\mathrm{Tr}\left[\Lambda \, \sigma_{Z^nJ_X^{l_x}J_Y^{l_y}}^{x_{m_{X_1}k_X}, y_{m_{Y_1}k_Y}, \theta^{l_x}, \phi^{l_y}}\right]},
\end{align}
where $\Lambda \triangleq \Lambda^{x_{m_{X_1}k_X}, y_{m_{Y_1}k_Y}}_{Z^n}\otimes \mathrm{id}_{J_XJ_Y}$. By the Gentle Measurement Lemma and \eqref{eq:general_success_prob}, we have for every $\theta^{l_x}, \phi^{l_y}$:
\begin{align}
    \frac{1}{2}\left\|\sigma - \sigma'\right\|_1 \leq 2\sqrt{46\varepsilon'_1}. \label{eq:layer2_sigma_distance}
\end{align}

Conditioned on correct Layer 1 decoding, Bob focuses on the relevant subsystems $Z^l J_X^{l_x} J_Y^{l_y}$. Applying Lemma~\ref{lemm:reliable_resolvable} to the Layer 2 superposed state:
\begin{align}
     \sigma^{x_{m_{X_1}k_X}, y_{m_{Y_1}k_Y}}_{\Theta_{X}^{l_x} \Theta_{Y}^{l_y}Z^lJ_X^{l_x}J_Y^{l_y}} 
     \triangleq \frac{1}{|\Theta_{X}|^{l_x}|\Theta_{Y}|^{l_y}}\sum_{\theta^{l_x}, \phi^{l_y}} \theta^{l_x}\otimes \phi^{l_y} \otimes \sigma^{x_{m_{X_1}k_X}, y_{m_{Y_1}k_Y}, \theta^{l_x}, \phi^{l_y}}_{Z^lJ_X^{l_x}J_Y^{l_y}}, \label{eq:layer_2_condi_state}
\end{align}
we can construct a codebook with average error probability at most $46\varepsilon_2$. However, decoding occurs on the post-measurement state $\sigma'$. Accounting for the perturbation \eqref{eq:layer2_sigma_distance}, the conditional error probability becomes:
\begin{align}
     &\mathbf{E}_{\mathcal{C}_{X_2}, \mathcal{C}_{Y_2}}\left[\mathbf{P}\left((\hat{M}_{X_2}, \hat{M}_{Y_2}) \neq (M_{X_2},M_{Y_2}) \mid (\hat{m}_{1}) = (m_{1})\right)\right] \nonumber \\
     &\leq 46\varepsilon_2 + 2\sqrt{46\varepsilon_1'}. \label{eq:layer_2_error_condi}
\end{align}
Averaging over all codebooks and utilizing Jensen's inequality:
\begin{align}
    &\mathbf{E}_{\mathcal{C}}\mathbf{E}_{M_1, S} \left[\mathbf{P}\left((\hat{M}_{X_2}, \hat{M}_{Y_2}) \neq (M_{X_2},M_{Y_2}) \mid (\hat{M}_{1}) = (M_{1})\right)\right] \nonumber \\
    &\leq \mathbf{E}_{\mathcal{C}}\mathbf{E}_{M_1, S} \left[46\varepsilon_2 + 2\sqrt{46\varepsilon_1'}\right] \leq 46\varepsilon_2 + 2\sqrt{46\varepsilon_1}.
\end{align}
Thus, the total probability of error is bounded by:
\begin{align}
    & \mathbf{E}\left[\mathbf{P}\left((\hat{M}_{1}) \neq (M_{1}) \vee (\hat{M}_{2}) \neq (M_{2})\right)\right] \nonumber \\
    & \leq \mathbf{E}\bigg[\mathbf{P}\left((\hat{M}_{1}) \neq (M_{1})\right) + \mathbf{P}\left((\hat{M}_{2}) \neq (M_{2})\mid (\hat{M}_{1}) = (M_{1})\right)\bigg] \nonumber \\
    &\leq 46\varepsilon_1 + 46\varepsilon_2 + 2\sqrt{46\varepsilon_1},\label{eq:error}
\end{align}
where the expectation $\mathbf{E}$ is taken over all codebooks $\mathcal{C}_{X_1}, \mathcal{C}_{Y_1}, \mathcal{C}_{X_2}, \mathcal{C}_{Y_2}$, keys, and messages.

\subsection{Throughput Analysis}\label{sec:throughput}
\subsubsection{Layer 1}
We begin by analyzing the throughput of Layer 1. Consider the average projected c-q state $\tilde\sigma_{X^nY^nZ^nW^n}$ defined via \eqref{eq:projected_state}. For this state, we apply the one-shot bounds from Lemma~\ref{lemm:reliable_resolvable}. The sizes of the message and secret key sets are bounded as follows:
\begin{align}
    &\log M_{X_1} = I_H^{(\varepsilon_1-\eta'')^2}(X^n:Y^nZ^n)_{\tilde{\sigma}} - 3 \log_2\left(\frac{46}{3}\eta''\right),\\
    & \log M_{Y_1} = I_H^{(\varepsilon_1-\eta'')^2}(Y^n:X^nZ^n)_{\tilde{\sigma}} - 3 \log_2\left(\frac{46}{3}\eta''\right),\\
    & \log M_{X_1} + \log M_{Y_1} = I_H^{(\varepsilon_1-\eta'')^2}(X^nY^n:Z^n)_{\tilde{\sigma}} - 3 \log_2\left(\frac{46}{3}\eta''\right),\label{eq:log_M}
\end{align}
and for the secret keys:
\begin{align}
    &\log M_{X_1}S_X = I^{\frac{\delta''-\delta'''}{6}}_{\max}(X^n:W^n)_{\rho} - 2 \log\left(\frac{\delta'''}{3}\right),\\
    &\log M_{Y_1}S_Y = I^{\frac{\delta''-\delta'''}{6}}_{\max}(Y^n:W^n)_{\rho} - 2 \log\left(\frac{\delta'''}{3}\right),\\
    & \log M_{X_1}S_X + \log M_{Y_1}S_Y = I^{\frac{\delta''-\delta'''}{12}}_{\max}(X^nY^n:W^n)_{\rho} - 2 \log\left(\frac{\delta'''}{6}\right).\label{eq:log_MK}
\end{align}

Furthermore, the covertness requirement imposed by Lemma~\ref{lemma:covertstate} (Section~\ref{sec:covert}) constrains the effective signal power. Specifically, using \eqref{eq:covert_mid_constraint}, we have:
\begin{align}
    &(\alpha_n\tau + \beta_n(1-\tau))s_n \nonumber\\
    &\leq \frac{2\sqrt{\kappa N_B(1+ \kappa N_B)}}{1-\kappa}Q^{-1}\left(\frac{1- (\delta-\eta)}{2}\right)n^{-\frac{1}{2}} \nonumber\\
    &\quad - \frac{\bar{D}(s_n^2+\alpha_n+\beta_n)}{\sqrt{n}}- O(n^{-1}),\\
    &= \frac{2\sqrt{\kappa N_B(1+ \kappa N_B)}}{1-\kappa}Q^{-1}\left(\frac{1- (\delta-\eta)}{2}\right)n^{-\frac{1}{2}} - o(n^{-\frac{1}{2}}).\label{eq:sparse_signal_limit}
\end{align}
For any $\eta>0$, there exists a sufficiently large $N_\eta$ such that \eqref{eq:sparse_signal_limit} holds for all $n \geq N_\eta$. Consequently, we obtain the inequality \eqref{eq:covert_constraint} stated in the main theorem:
\begin{align}
    &(\alpha_n\tau + \beta_n(1-\tau))s_n\\
    &\leq \frac{2\sqrt{\kappa N_B(1+ \kappa N_B)}}{1-\kappa}\left(Q^{-1}\left(\frac{1- \delta}{2}\right) - o(1)\right)n^{-\frac{1}{2}} - o(n^{-\frac{1}{2}})\nonumber\\
    &= \frac{2\sqrt{\kappa N_B(1+ \kappa N_B)}}{1-\kappa}Q^{-1}\left(\frac{1- \delta}{2}\right) n^{-\frac{1}{2}} - o(n^{-\frac{1}{2}}).
\end{align}

To expand the information quantities, recall from \eqref{eq:projection_layer_1} that for each $n$, we can choose a projection such that the trace distance is arbitrarily small:
\begin{align}
     \frac{1}{2}\left\|\tilde\sigma_{X^nY^nB^n} - \tilde\sigma_{X^nY^nZ^n}\right\|_1 \leq \delta_n. 
\end{align}
We set $\delta_n \leq \exp(-\omega\sqrt{n})$ and apply Lemma~\ref{lemma:typical_product} to the state $\tilde\sigma_{X^nY^nB^n}$. the two prerequisites for Lemma~\ref{lemm:expansion} are thus satisfied:
\begin{align}
    &\frac{1}{2}\left\|\tilde\sigma_{X^nY^nZ^n} - \tilde\sigma_{X^nY^nB^n}\right\|_1 \leq \delta_n \leq \exp(-\omega(\sqrt{n})), \\
    &\frac{1}{2}\left\|\tilde\sigma_{X^nY^nB^n} - \sigma_{XYB}^{\otimes n}\right\|_1 \leq 2\exp\left(-\frac{1}{2}\bar\mu^2n\alpha_n\right) + 2\exp\left(-\frac{1}{2}\bar\mu^2n\beta_n\right).
\end{align}
Setting $\eta'' = \frac{1}{2\varepsilon_1\sqrt{n}}$ in \eqref{eq:log_M} and invoking Lemma~\ref{lemm:expansion}, we obtain the second-order expansion for the sum rate:
\begin{align}
    \log M_{X_1} + \log M_{Y_1} 
    &\geq nD(\sigma_{XYB}\|\sigma_{XY}\otimes \sigma_{B}) - \sqrt{nV(\sigma_{XYB}\|\sigma_{XY}\otimes \sigma_{B})}Q^{-1}(\varepsilon_1^2)\nonumber\\
    &\quad+ O\left(\frac{R(\sigma_{XYB}\|\sigma_{XY}\otimes \sigma_{B})^\frac{3}{4}}{V(\sigma_{XYB}\|\sigma_{XY}\otimes \sigma_{B})}\right) + O(\log n).
\end{align}
Applying the throughput bounds from Lemma~\ref{lemma:throughput}, this simplifies to:
\begin{align}
    \log M_{X_1} + \log M_{Y_1} 
    &\geq \frac{n(\alpha_n\tau^2+\beta_n(1-\tau)^2)\kappa^2 s_n^2}{2(1-\kappa)N_B(1+ (1-\kappa)N_B)} + o(n(\alpha_n+ \beta_n)s_n^2).
\end{align}
Similarly, the individual rates are bounded by:
\begin{align}
    \log M_{X_1} &\geq \frac{n\alpha_n\tau^2\kappa^2 s_n^2}{2(1-\kappa)N_B(1+ (1-\kappa)N_B)} + o(n\alpha_ns_n^2),\\
    \log M_{Y_1} &\geq \frac{n\beta_n(1-\tau)^2\kappa^2 s_n^2}{2(1-\kappa)N_B(1+ (1-\kappa)N_B)} + o(n\beta_ns_n^2).
\end{align}
These inequalities correspond to the achievable Layer 1 message rates asserted in Theorem~\ref{thm:main}.

Finally, we address the resolvability rates. Applying Lemma~\ref{lemma:typical_product} to $\tilde\sigma_{X^nY^nW^n}$, we confirm the prerequisite for the second part of Lemma~\ref{lemm:expansion}:
\begin{align}
    \frac{1}{2}\left\|\tilde\sigma_{X^nY^nW^n} - \sigma_{XYW}^{\otimes n}\right\|_1 \leq \exp\left(-\frac{1}{2}\bar\mu^2n\alpha_n\right) + \exp\left(-\frac{1}{2}\bar\mu^2n\beta_n\right).
\end{align}
Setting $\delta'''= \frac{12}{\sqrt{n}}$ in \eqref{eq:log_MK}, we apply Lemma~\ref{lemm:expansion} to obtain:
\begin{align}
    &\log M_{X_1}S_X + \log M_{Y_1}S_Y \nonumber\\
    &\leq nD(\sigma_{XYW}\|\sigma_{XY}\otimes \sigma_{W}) + \sqrt{nV(\sigma_{XYW}\|\sigma_{XY}\otimes \sigma_{W})}Q^{-1}\left(\frac{\delta''^2}{144}\right)\nonumber\\
    &\quad + O\left(\frac{R(\sigma_{XYW}\|\sigma_{XY}\otimes \sigma_{W})^\frac{3}{4}}{V(\sigma_{XYW}\|\sigma_{XY}\otimes \sigma_{W})}\right).
\end{align}
Applying Lemma~\ref{lemma:throughput}, the sum key rate is bounded by:
\begin{align}
    &\log M_{X_1}S_X + \log M_{Y_1}S_Y \nonumber\\
    &\leq \frac{n(\alpha_n\tau^2+\beta_n(1-\tau)^2)(1-\kappa)^2 s_n^2}{2\kappa N_B(1+ \kappa N_B)} + o(n(\alpha_n+ \beta_n)s_n^2).
\end{align}
Similarly, the individual key rates satisfy:
\begin{align}
    &\log M_{X_1}S_X \leq \frac{n\alpha_n\tau^2(1-\kappa)^2 s_n^2}{2\kappa N_B(1+ \kappa N_B)} + o(n\alpha_ns_n^2),\\
    &\log M_{Y_1}S_Y \leq \frac{n\beta_n(1-\tau)^2(1-\kappa)^2 s_n^2}{2\kappa N_B(1+ \kappa N_B)} + o(n\beta_ns_n^2).
\end{align}
These inequalities confirm the sufficient resolvability rates required by Theorem~\ref{thm:main}.

\subsubsection{Layer 2}\label{sec:throughput_layer_2}
We now analyze Layer 2, conditioning on the successful decoding of Layer 1 messages $(x_{m_{X_1}k_X}, y_{m_{Y_1}k_Y})$. Let $l$ denote the number of relevant $Z$-systems correlated with the idlers $J_X, J_Y$. The corresponding conditional state for Layer 2 is:
\begin{align}
     &\sigma^{x_{m_{X_1}k_X}, y_{m_{Y_1}k_Y}}_{\Theta_{X}^{l_x} \Theta_{Y}^{l_y}Z^lJ_X^{l_x}J_Y^{l_y}}\notag \\
     &= \frac{1}{|\Theta_{X}|^{l_x}|\Theta_{Y}|^{l_y}}\sum_{\theta^{l_x}, \phi^{l_y}} \ket{\theta^{l_x}}\bra{\theta^{l_x}}\otimes \ket{\phi^{l_y}}\bra{\phi^{l_y}} \otimes \sigma^{x_{m_{X_1}k_X}, y_{m_{Y_1}k_Y}, \theta^{l_x}, \phi^{l_y}}_{Z^lJ_X^{l_x}J_Y^{l_y}}.
\end{align}
Applying Lemma~\ref{lemm:reliable_resolvable} yields the following one-shot lower bounds on the Layer 2 codebook sizes for any $\eta \in (0, \varepsilon_2)$:
\begin{align}
    \log M_{X_2} &\geq I_H^{(\varepsilon_2-\eta)^2}(\Theta_{X}^{l_x}:\Theta_{Y}^{l_y}Z^lJ_X^{l_x}J_Y^{l_y})_{\sigma} - 3 \log_2\left(\frac{46}{3}\eta\right),\\
    \log M_{Y_2} &\geq I_H^{(\varepsilon_2-\eta)^2}(\Theta_{Y}^{l_y}:\Theta_{X}^{l_x}Z^lJ_X^{l_x}J_Y^{l_y})_{\sigma} - 3 \log_2\left(\frac{46}{3}\eta\right),\\
    \log M_{X_2}M_{Y_2} &\geq I_H^{(\varepsilon_2-\eta)^2}(\Theta_{X}^{l_x}\Theta_{Y}^{l_y}:Z^lJ_X^{l_x}J_Y^{l_y})_{\sigma} - 3 \log_2\left(\frac{46}{3}\eta\right).
\end{align}
Recall that the projection error in \eqref{eq:projection_layer_2} is bounded by $\delta_n$. By selecting $\delta_n$ sufficiently small and setting $\eta = O(n^{-1/2})$, we employ the second-order expansion (Lemma~\ref{lemm:expansion}) to recover the infinite-dimensional mutual information quantities:
\begin{align}
    \log M_{X_2} &\geq I(\Theta_{X}^{l_x}:\Theta_{Y}^{l_y}B^lJ_X^{l_x}J_Y^{l_y})_{\sigma} - \sqrt{V_{\sigma}} Q^{-1}(\varepsilon_2^2) + h.o.t.,
\end{align}
with analogous expansions for $\log M_{Y_2}$ and the sum rate. In the covert regime, the finite blocklength penalties are asymptotically negligible compared to the dominant $O(\sqrt{n} \log n)$ scaling of the mutual information. Consequently, the achievable rate is governed by the first-order term.

With Lemma~\ref{lemma_D_theta} below, we con substitute the explicit expansions for the mutual information.
\begin{lemma} \label{lemma_D_theta} 
    The mutual information terms satisfy:
    \begin{align}           
    I(\Theta_{X}^{l_x}:\Theta_{Y}^{l_y}B^lI_X^{l_x}I_Y^{l_y})_{\sigma} &= -(1-\bar\mu)\frac{n\kappa\tau\alpha_n}{1+(1-\kappa)N_B}s_n\log s_n + O(n\alpha_ns_n),\label{eq:theta_X}\\
    I(\Theta_{Y}^{l_y}:\Theta_{X}^{l_x}B^lI_X^{l_x}I_Y^{l_y})_{\sigma} &= -(1-\bar\mu)\frac{n\kappa (1-\tau)\beta_n}{1+(1-\kappa)N_B}s_n\log s_n + O(n\beta_ns_n),\label{eq:theta_Y}\\
    I(\Theta_{X}^{l_x}\Theta_{Y}^{l_y}:B^lI_X^{l_x}I_Y^{l_y})_{\sigma} &= -(1-\bar\mu)\frac{n\kappa(\tau\alpha_n + (1-\tau)\beta_n)}{1+(1-\kappa)N_B}s_n\log s_n \notag\\
    & \quad + O(n(\alpha_n+\beta_n)s_n).\label{eq:theta_XY}
    \end{align}
\end{lemma}
\begin{proof}
Note that $\alpha_n, \beta_n, s_n = o(1) = \omega(n^{-\frac{1}{2}})$ and $(\alpha_n+\beta_n)s_n \leq O(n^{-\frac{1}{2}})$, 
and that we have the following expressions of relative entropy shown in Lemma~\ref{lemma_D_theta_1} or in similar fashion:
\begin{align}
    &D\left(\sigma_{\Theta_{X}\Theta_{Y}BI_XI_Y} \right\|\left.\sigma_{\Theta_{X}\Theta_{Y}}\otimes \sigma_{BI_XI_Y}\right) = -\frac{\kappa s_n\log s_n}{1+ (1-\kappa)N_B} + O(s_n)\\
    & D\left(\sigma_{\Theta_{X}\Theta_{Y}BI_XI_Y} \right\|\left.\sigma_{\Theta_{X}}\otimes \sigma_{\Theta_{Y}BI_XI_Y}\right)  \notag\\
    & \approx D\left(\bar\sigma_{\Theta_{X}BI_X} \right\|\left.\bar\sigma_{\Theta_{X}}\otimes \bar\sigma_{BI_X}\right) 
     = -\frac{\kappa\tau s_n\log s_n}{1+ (1-\kappa)N_B} + O(s_n);\\
     & D\left(\sigma_{\Theta_{X}\Theta_{Y}BI_XI_Y} \right\|\left.\sigma_{\Theta_{Y}}\otimes \sigma_{\Theta_{X}BI_XI_Y}\right)  \notag\\
    &\approx D\left(\bar\sigma_{\Theta_{Y}BI_Y} \right\|\left.\bar\sigma_{\Theta_{Y}}\otimes \bar\sigma_{BI_Y}\right) 
     =  -\frac{\kappa(1-\tau) s_n\log s_n}{1+ (1-\kappa)N_B} + O(s_n).
\end{align}
Here $\approx$ means equal up to order $O(s_n)$, and $\bar\sigma_{\Theta_{X} B I_X}$ (resp. $\bar\sigma_{\Theta_{Y} B I_Y}$)  describes the state when only $A_1$ (resp. $A_2$) transmits the message $\theta$ to B while $A_2$ (resp. $A_1$) does not. 
Denote $\tilde{l} \triangleq l_x + l_y - l$ the number of channels where both senders transmit simultaneously.
For \eqref{eq:theta_XY}, we obtain
\begin{align}
    &D\left(\sigma_{\Theta_{X}^{l_x}\Theta_{Y}^{l_y}B^lI_X^{l_x}I_Y^{l_y}} \right\|\left.\sigma_{\Theta_{X}^{l_x}\Theta_{Y}^{l_y}}\otimes \sigma_{B^lI_X^{l_x}I_Y^{l_y}}\right) \notag\\
    & = (l_x - \tilde{l})D\left(\bar\sigma_{\Theta_{X}BI_X} \right\|\left.\bar\sigma_{\Theta_{X}}\otimes \bar\sigma_{BI_X}\right) + (l_y - \tilde{l})D\left(\bar\sigma_{\Theta_{Y}BI_Y} \right\|\left.\bar\sigma_{\Theta_{Y}}\otimes \bar\sigma_{BI_Y}\right)\notag\\
    &\quad + \tilde{l}D\left(\sigma_{\Theta_{X}\Theta_{Y}BI_XI_Y} \right\|\left.\sigma_{\Theta_{X}\Theta_{Y}}\otimes \sigma_{BI_XI_Y}\right)\notag\\
    & = (1-\bar\mu) \frac{-n\kappa(\tau\alpha_n + (1-\tau)\beta_n)s_n\log s_n}{1+(1-\kappa)N_B} + O(n(\alpha_n+\beta_n)s_n)\notag\\
    & =  (1-\bar\mu) \frac{-n\kappa(\tau\alpha_n + (1-\tau)\beta_n)s_n\log s_n}{1+(1-\kappa)N_B} - o(\sqrt{n} \log n).
\end{align}
As for \eqref{eq:theta_X}, we have
\begin{align}
    &D\left(\sigma_{\Theta_{X}^{l_x}\Theta_{Y}^{l_y}B^lI_X^{l_x}I_Y^{l_y}} \right\|\left.\sigma_{\Theta_{X}^{l_x}}\otimes \sigma_{\Theta_{Y}^{l_y}B^lI_X^{l_x}I_Y^{l_y}}\right) \notag\\
    & = (l_x - \tilde{l})D\left(\bar\sigma_{\Theta_{X}BI_X} \right\|\left.\bar\sigma_{\Theta_{X}}\otimes \bar\sigma_{BI_X}\right) + \tilde{l}D\left(\sigma_{\Theta_{X}\Theta_{Y}BI_XI_Y} \right\|\left.\sigma_{\Theta_{X}}\otimes \sigma_{\Theta_{Y} BI_XI_Y}\right)\notag\\
    & = (1-\bar\mu) \frac{-n\kappa\tau\alpha_n s_n\log s_n}{1+(1-\kappa)N_B} + O(n\alpha_ns_n).
\end{align}
Equation \eqref{eq:theta_Y} can be bounded similarly and the proof is done.
\end{proof}
Combining these expansions with the rate bounds, we obtain the final Layer 2 rate guarantees stated in Theorem~\ref{thm:main}:
\begin{align}           
    \log M_{X_2} &\geq -(1-\bar\mu)\frac{n\kappa\tau\alpha_n}{1+(1-\kappa)N_B}s_n\log s_n + o(\sqrt{n}\log n),\\
    \log M_{Y_2} &\geq -(1-\bar\mu)\frac{n\kappa (1-\tau)\beta_n}{1+(1-\kappa)N_B}s_n\log s_n + o(\sqrt{n}\log n),\\
    \log M_{X_2}M_{Y_2} &\geq -(1-\bar\mu)\frac{n\kappa(\tau\alpha_n + (1-\tau)\beta_n)}{1+(1-\kappa)N_B}s_n\log s_n + o(\sqrt{n}\log n).
\end{align}
\subsection{Existence Argument}
To establish the existence of a valid coding scheme, we first bound the ensemble-average error probability. Recalling the total error bound derived in~\eqref{eq:error}, we have:
\begin{align}
    &\mathbf{E}_{\mathcal{C}_{X_1}, \mathcal{C}_{Y_1}, \mathcal{C}_{X_2}, \mathcal{C}_{Y_2}} \mathbf{E}_{S_X, S_Y} \left[
    \mathbf{P}\left(\hat{M}_1 \neq M_1 \vee \hat{M}_2 \neq M_2\right)
    \right] \notag\\
    &\quad\leq 46\varepsilon_1 + 46\varepsilon_2 + 2\sqrt{46\varepsilon_1}.
\end{align}
To demonstrate the existence of a codebook tuple $(\mathcal{C}_{X_1}, \mathcal{C}_{Y_1}, \mathcal{C}_{X_2}, \mathcal{C}_{Y_2})$ that is $\varepsilon$-reliable, we require the expected error probability over the keys to satisfy \eqref{eq:reliable_equa}:
\begin{align}
    P_e = \mathbf{E}_{S_X, S_Y} \left[
    \mathbf{P}\left(\hat{M}_1 \neq M_1 \vee \hat{M}_2 \neq M_2\right)
    \right] \leq \varepsilon. \label{eq:error_bound}
\end{align}
This condition is satisfied by choosing the parameters $\varepsilon_1$ and $\varepsilon_2$ sufficiently small and taking the blocklength $n$ large enough such that the bound becomes negligible compared to the target error:
\begin{align}
    46\varepsilon_1 + 46\varepsilon_2 + 2\sqrt{46\varepsilon_1} \ll \varepsilon.
\end{align}
By applying Markov's inequality, we deduce that a randomly generated codebook tuple satisfies the reliability condition~\eqref{eq:error_bound} with high probability.

Simultaneously, regarding the covertness constraint, the analysis in Section~\ref{sec:covert} establishes that a randomly selected Layer 1 codebook tuple $(\mathcal{C}_{X_1}, \mathcal{C}_{Y_1})$ is $\delta$-covert (satisfying \eqref{eq:covert_equa}) with probability approaching one, provided that $\delta'' \ll \eta - \eta'$.

Finally, by the union bound, the probability that a random codebook fails either the reliability or the covertness criterion can be made strictly less than one. Consequently, the intersection of the sets of reliable and covert codebooks is non-empty, proving the existence of a tuple $(\mathcal{C}_{X_1}, \mathcal{C}_{Y_1}, \mathcal{C}_{X_2}, \mathcal{C}_{Y_2})$ that is simultaneously $\varepsilon$-reliable and $\delta$-covert.

\section*{Acknowledgment}
Parts of this document have received assistance from generative AI tools to aid in the composition; the authors have reviewed and edited the content as needed and take full responsibility for it.

\bibliographystyle{IEEEtran}
\bibliography{references.bib}

\end{document}